\documentclass[12pt]{article}

\usepackage[letterpaper,margin=1in]{geometry}
\usepackage[T1]{fontenc}
\usepackage[utf8]{inputenc}
\usepackage{lmodern}
\usepackage{microtype}
\usepackage{amsmath,amssymb,amsfonts,amsthm,mathtools,bm}
\usepackage{booktabs,array,threeparttable}
\usepackage{enumitem}
\usepackage{setspace}
\usepackage[dvipsnames]{xcolor}
\usepackage{graphicx}
\usepackage{caption}
\usepackage{tikz}
\usetikzlibrary{arrows.meta,positioning,calc}
\usepackage[authoryear,round,longnamesfirst]{natbib}
\usepackage[colorlinks=true,citecolor=black,urlcolor=black,linkcolor=black]{hyperref}
\usepackage[nameinlink,noabbrev]{cleveref}
\usepackage{pgfplots}
\usepgfplotslibrary{groupplots}
\pgfplotsset{compat=1.18}

\newcommand{\manuscriptspacing}{\setstretch{1.45}}
\allowdisplaybreaks
\newtheorem{theorem}{Theorem}
\newtheorem{proposition}{Proposition}
\newtheorem{lemma}{Lemma}
\newtheorem{corollary}{Corollary}
\newtheorem{assumption}{Assumption}
\newtheorem{definition}{Definition}
\newtheorem{remark}{Remark}
\newtheorem{example}{Example}

\crefname{theorem}{theorem}{theorems}
\Crefname{theorem}{Theorem}{Theorems}
\crefname{proposition}{proposition}{propositions}
\Crefname{proposition}{Proposition}{Propositions}
\crefname{lemma}{lemma}{lemmas}
\Crefname{lemma}{Lemma}{Lemmas}
\crefname{corollary}{corollary}{corollaries}
\Crefname{corollary}{Corollary}{Corollaries}
\crefname{assumption}{assumption}{assumptions}
\Crefname{assumption}{Assumption}{Assumptions}
\crefname{definition}{definition}{definitions}
\Crefname{definition}{Definition}{Definitions}
\crefname{remark}{remark}{remarks}
\Crefname{remark}{Remark}{Remarks}
\crefname{example}{example}{examples}
\Crefname{example}{Example}{Examples}

\DeclareMathOperator*{\argmin}{arg\,min}
\DeclareMathOperator*{\argmax}{arg\,max}
\newcommand{\E}{\mathbb E}
\newcommand{\V}{\mathbb V}
\newcommand{\Pp}{\mathbb P}
\newcommand{\1}{\mathbf 1}
\newcommand{\cA}{\mathcal A}
\newcommand{\cG}{\mathcal G}
\newcommand{\cS}{\mathcal S}
\newcommand{\cX}{\mathcal X}
\newcommand{\cI}{\mathcal I}
\newcommand{\RA}{\mathrm{RA}}
\newcommand{\PM}{\mathrm{PM}}
\newcommand{\RR}{\mathrm{R}}
\newcommand{\tr}{\mathrm{tr}}
\newcommand{\te}{\mathrm{te}}
\newcommand{\CE}{\mathrm{CE}}
\newcommand{\TV}{\mathrm{TV}}
\newcommand{\KL}{\mathrm{KL}}
\newcommand{\dd}{\,\mathrm d}
\newcommand{\appcontentsline}[2]{%
  \noindent\hyperref[#1]{#2}\dotfill\pageref{#1}\par}

\title{\textbf{Stable Policy Learning}}
\author{%
Harvey Barnhard\\[-1pt]
{\normalsize Harvard University}
\and
Giacomo Opocher\\[-1pt]
{\normalsize University of Mannheim}
\and
Rahul Singh\thanks{
    We thank Amartya Sen and Davide Viviano for helpful discussions.}\hspace{.2cm}\\[-1pt]
{\normalsize Harvard University}}
\date{September 16, 2026}

\begin{document}
\maketitle

\begin{abstract}
In evidence-based policymaking, typically one experimental sample is observed, then a learned policy recommendation is implemented at scale.
Policies learned from the experimental data can perform well in expected welfare, yet random sampling in the experiment can produce recommendations with poor welfare outcomes.
In this paper, we ask: how should policy learning algorithms balance expected welfare against sampling risk?
Our main contribution is to show that algorithmic stability plays a central role in characterizing and navigating the tradeoff.
Intuitively, if a policy learning algorithm's recommendation remains stable when one experimental unit is replaced, then that algorithm has limited sampling risk.
We propose a method for policy learning called policy-vote bagging, which learns treatment decisions on many subsamples then averages their votes into treatment probabilities. 
Relative to using one subsample, averaging across subsamples preserves expected welfare and improves expected utility for a risk-averse researcher.
We derive sharp bounds linking estimation accuracy, subsample size, and welfare variation, including an exact guarantee under CARA utility.
\end{abstract}

\vfill

\clearpage
\manuscriptspacing

\section{Introduction}\label{sec:introduction}

Many economic policies are learned from experiments and then implemented at scale. 
Often, one experimental sample is observed, one resulting policy is implemented, and a poor policy recommendation may be difficult to reverse.
Standard treatment-choice frameworks evaluate an algorithm by its expected welfare across hypothetical experimental samples. 
This criterion does not distinguish between two algorithms that perform equally well in expectation, even though one algorithm may produce much more variable recommendations across hypothetical experiments. 
The distinction matters when researchers, institutions, or societies are averse to the inherent gamble of learning a policy from randomly sampled experimental data.

This paper views the choice of a policy learning algorithm as part of the research design. Before observing the experimental sample, a researcher chooses how the data will be converted into a policy.\footnote{Choosing an algorithm before observing the data
follows the statistical decision-theoretic framework of
\citet{wald1950}. A decision rule maps possible
samples into actions. Its performance is evaluated
across samples.} The experimental sample is then collected, the algorithm recommends a policy, and that policy determines treatment in the full population. This sequence creates two sources of uncertainty in chronological order, which may be viewed as a compound lottery. First, the learned policy varies across possible experimental samples. Second, outcomes vary across people in the full population. The main text focuses on aversion to the first source of variation. Appendix~\ref{app:joint_concavity} allows aversion to both.\footnote{This compound lottery extends the veil of ignorance  \citep{vickrey_1945,harsanyi_1953,rawls_1971} to before the experiment. The researcher chooses how to learn a policy without knowing which sample will be observed. The policymaker evaluates a policy without knowing whose outcome they will receive. }

Our primary contribution is to prove that algorithmic stability provides a general way to manage experimental sampling risk. We say that a policy learning algorithm is welfare stable when replacing one experimental unit changes the welfare of its learned policy by a small amount. This one-unit stability controls both welfare variation across samples and the probability that the learned policy performs poorly. It also yields an exact, finite-sample guarantee for a researcher with constant absolute risk aversion. The result is deliberately broad: it requires only that welfare be well defined, and places no restriction on the dimension of the covariates nor the complexity of the policy class.

Our secondary contribution is to propose a method for stable policy learning called  policy-vote bagging. The method learns a binary treatment decision on each of many subsamples and uses the fraction voting for treatment as the final treatment probability. Its stability comes directly from the construction. If one unit is replaced, every subsample that omits that unit casts the same vote as before; only the fraction of subsamples containing it can change. Policy-vote bagging therefore achieves stability across the full experimental sample, even when the underlying base learner is unstable. At a fixed subsample size, averaging preserves expected welfare (relative to a policy based on one subsample) and improves expected utility for any increasing, concave evaluation of welfare.

Our formal results show how subsample size trades estimation accuracy for stability. Larger subsamples give the base learner more information, but they also make each unit influential in more votes. We derive matching upper and lower bounds over a class of populations and base learners described by their mean-squared estimation error. Within this class, the bounds characterize a  minimax rate-optimal subsample size. Parametric, H\"older-smooth, and Gaussian-kernel examples show how familiar estimation guarantees enter the analysis. 
Subsampling can be worthwhile in small experimental samples when estimation accuracy improves slowly in the sample size and sampling risk is prioritized.

\subsection{Related work}

We contribute to the literature on statistical treatment choice; see \citet{opocher_26_review} for a recent review. Expected-regret analysis \citep{manski2004} underlies empirical welfare maximization and its extensions \citep{kitagawa_tetenov_2018,mbakop_tabord_meehan_2021,athey_wager_2021}. Related work incorporates concerns about inequality and fairness into the objectives used to evaluate treatment assignments \citep{kitagawa_tetenov_2021,viviano_bradic_2024}, yielding different criteria; see Section~\ref{sec:compound}.

Our closest connection is to work that evaluates welfare risk across experimental samples. In an innovation experiment with binary outcomes, \citet{manski_tetenov_2007} evaluate treatment rules through the expectation of a concave transformation of population welfare and characterize admissible fractional treatment rules. Their criterion inspires our formalization of sampling risk. \citet{manski_tetenov_2023} study quantile performance and stochastic-dominance comparisons of sampling distributions, while \citet{kitagawa_lee_qiu_2026} study nonlinear transformations of regret, discussed in Section~\ref{sec:compound}. We connect sampling risk to algorithmic stability and derive quantitative  guarantees for policy-vote bagging with covariates.

Other approaches account for statistical uncertainty when
choosing policies or reporting their performance.
These include policy choice with budget constraints and estimated
costs \citep{sun2026}, policy choice based on confidence
bounds for welfare
\citep{chernozhukov_lee_rosen_sun_2025,andrews_chen_2025},
and empirical Bayes allocation across policies with estimated
costs and benefits \citep{moon_2026}.
\citet{swaminathan_joachims_2015} account for uncertainty
in an estimator of welfare when choosing a policy.
By contrast, our criterion concerns how the actual welfare
of the learned policy varies across experimental samples.
We show how algorithmic stability controls this variation.

Our two-stage formulation resembles models of ambiguity aversion. Smooth, variational, multiplier, and recursive preferences distinguish layers of uncertainty \citep{klibanoff_et_al_2005,maccheroni_marinacci_rustichini_2006,strzalecki_2011,strzalecki_2013,denti_pomatto_2022}. Robust mean-variance analysis provides a particularly close mathematical analogy \citep{maccheroni_marinacci_ruffino_2013}. Our interpretation is different. Both the experimental-sample lottery and full-population lottery are objective and arise under a fixed data-generating process.

Single-unit stability has separate roots in econometrics and machine learning. Building on the influence framework of \citet{hampel_1974}, \citet{andrews_1986} formalizes a stability exponent based on the rate at which an estimator's largest leave-one-out change vanishes. Modern algorithmic-stability analyses use this property to obtain finite-sample statistical guarantees \citep{bousquet_elisseeff_2002}. \citet{elisseeff_evgeniou_pontil_2005} extend this framework to randomized learners, and analyze bagging and subbagging. In recent work on semiparametric inference, \citet{chernozhukov_newey_singh_syrgkanis_2026} and \citet{chen_syrgkanis_austern_2022} develop a stability-based analysis and show that subsample bagging can provide the required stability. Our definition of welfare stability shares this single-unit logic. Unlike earlier works, we apply it to the population welfare of a learned policy, and use it to control variation across experimental samples.

Bagging can stabilize an unstable base learner. Classical work studies its effects on instability and prediction error \citep{breiman_1996,buhlmann_yu_2002,buja_stuetzle_2006}. \citet{soloff_barber_willett_2024} give assumption-free stability guarantees for bagged predictors. Stable predictions need not yield stable decisions, so  \citet{soloff_barber_willett_argmax_2024} address discontinuous classification by reporting sets of candidate labels, and \citet{qian_ying_lam_yin_2025} instead select models through subsample voting. Our proposal of policy-vote bagging retains the vote shares as treatment probabilities. We study how this fractional implementation controls the sampling distribution of population welfare.

The paper follows the decision problem from criterion to method. \Cref{sec:compound} defines the two-stage lottery and risk-aware regret. \Cref{sec:design} interprets algorithm choice as a tradeoff between expected welfare and sampling risk. \Cref{sec:stability} shows how stability controls experimental-sample risk. \Cref{sec:bagging} develops policy-vote bagging. \Cref{sec:rates} gives lower bounds and optimality results. \Cref{sec:discussion} concludes. Appendix~\ref{app:joint_concavity} allows risk aversion at both stages.  Appendices~\ref{app:sec2_proofs}-\ref{app:secA_proofs} collect the proofs.
\section{The training-test compound lottery}\label{sec:compound}

This section separates the policy-learning problem into its two stages. We define the experimental sample used to learn a policy, and the welfare in the full population. We place them in a chronological lottery. As shorthand, we refer to the former source of uncertainty as arising from a ``training'' set, and the latter as arising from a ``test'' set. Finally, we introduce risk-aware regret as a tractable criterion integrating both sources of uncertainty.

\subsection{Data setting}

We distinguish the population that receives the policy (the ``test'' set) from the experiment used to learn it (the ``training'' set). Let $P$ denote the joint law of covariates $X$ and potential outcomes $Y(0),Y(1)$ for a generic test unit, with $X\in\cX\subseteq\mathbb R^d$. The conditional average treatment effect is $\tau(x)=\E\{Y(1)-Y(0)\mid X=x\}$.

A policy specifies a treatment probability for each covariate value. Formally, $G:\cX\to[0,1]$ is measurable, and $\cG$ is the class of all such rules. This unconstrained class contains the oracle and the averages of policies studied below. For an independent test unit and $U^{\te}\sim\operatorname{Unif}(0,1)$, set $D^G=\1\{U^{\te}\leq G(X^{\te})\}$ and $Y^G=Y^{\te}(D^G)$. Thus, the test outcome reflects population heterogeneity and, when $G$ is fractional, the policy's own randomization.

The policymaker may value the distribution of test outcomes as well as their mean. Under policy $G$, realized utility is $U_{\PM}(G)=\phi(Y^G)$ and ex ante value is $V_\phi(G)=\E\{\phi(Y^G)\}$, where $\phi$ is increasing and concave, and $\E|\phi\{Y(d)\}|<\infty$ for $d\in\{0,1\}$. Concavity places less value on mean-preserving spreads in outcomes, which we call inequality aversion.\footnote{It also evaluates idiosyncratic outcome risk.} Conditional independence of the policy randomization gives
\[
V_\phi(G)
=
\E\bigl[G(X)\phi\{Y(1)\}+\{1-G(X)\}\phi\{Y(0)\}\bigr].
\]
The relevant treatment contrast is therefore $\tau_\phi(x)=\E[\phi\{Y(1)\}-\phi\{Y(0)\}\mid X=x]$, and the unconstrained oracle treats when this contrast is nonnegative: $G_\phi^\star(x)=\1\{\tau_\phi(x)\geq0\}$. 

The training data come from a randomized experiment, sampled from the same  population. The researcher observes $S=(S_1,\ldots,S_n)\sim P_{\tr}^n$, where $S_i=(X_i,D_i,Y_i)$ and $Y_i=D_iY_i(1)+(1-D_i)Y_i(0)$. The law $P_{\tr}$ draws covariates and potential outcomes from $P$ and then randomizes treatment within covariate strata. The following assumption makes this design explicit and imposes the regularity used below.

\begin{assumption}[Training sample and test outcome]\label[assumption]{ass:main}
The following conditions hold.
\begin{enumerate}[label=(\roman*),leftmargin=2.1em]
\item Integrability: $\E|Y(0)|+\E|Y(1)|<\infty$.
\item Stratified randomization: for the training experiment,
$
D\perp\{Y(0),Y(1)\}\mid X,
$ and $
\Pp(D=1\mid X=x)\in[\underline e,1-\underline e]
$
for some $\underline e\in(0,1/2)$ and all $x\in\cX$.
\item Margin condition: there exist $C_\tau<\infty$ and $\kappa>0$ such that, for all $t>0$,
$
\Pp_X\{0<|\tau(X)|\leq t\}\leq C_\tau t^\kappa.
$
\end{enumerate}
\end{assumption}

Each part of Assumption~\ref{ass:main} has a separate role. Integrability requires only first moments, rather than bounded outcomes, and ensures that welfare is well defined for every $G\in\cG$. Randomization identifies the CATE from the training sample. The margin condition limits how much of the test population is nearly indifferent between treatment and no treatment. 

\subsection{The researcher's decision problem}

The researcher's choice precedes both the training and test lotteries. A learning algorithm is a measurable map $A:\cS^n\to\cG$, and at date zero the researcher chooses $A$ from an admissible class $\cA$. The sequence of events is as follows. We abbreviate the researcher by R, and the policymaker by PM.

\begin{description}[leftmargin=5.2em,labelwidth=5.2em,style=nextline]
\item[$t=0$.] R decides on an algorithm $A\in\cA$ before observing the training data.
\item[Training.] The randomized training experiment produces $S\sim P_{\tr}^n$.
\item[$t=1$.] R recommends the data-driven policy $\widehat G=A(S)$, and the PM decrees it. At this date, R evaluates the recommendation using the PM's expected utility conditional upon the training sample.
\item[Test.] An independent test unit is assigned according to $\widehat G$, and only the selected potential outcome $Y^{\widehat G}$ realizes.
\item[$t=2$.] The PM receives utility $U_{\PM}(\widehat G)=\phi(Y^{\widehat G})$.
\end{description}

The two lotteries enter the researcher's objective in different places. Conditional on $S$, the policymaker values the learned policy as $\E_P[\phi\{Y^{A(S)}\}\mid S]=V_\phi\{A(S)\}$. The policymaker accepts the recommendation rather than choosing among algorithms. The researcher's date-one utility is $U_{\RR}(A)=\psi[V_\phi\{A(S)\}]$, where the increasing, concave function $\psi$ expresses aversion to variation in policy value across training samples. At date zero, R solves
\begin{equation}\label{eq:researcher_problem}
A_\star
\in
\argmax_{A\in\cA}
\E_{P_{\tr}^n}
\left\{
\psi\left(
\E_P\left[
\phi\left\{Y^{A(S)}\right\}
\mid S
\right]
\right)
\right\}.
\end{equation}

The criterion in \eqref{eq:researcher_problem} makes reliability across training samples part of statistical treatment choice. R can be interpreted as a benevolent adviser who dislikes exposing the population to a sampling gamble. The same reduced-form preference can represent costly policy reversal or reputational losses from a recommendation that performs poorly because of the realized sample. Alternatively, $\psi$ can be an ex ante social criterion delegated to R. The two-agent language simply keeps the training and test lotteries conceptually separate; it does not require their objectives to conflict.

\begin{remark}[Taxonomy]\label[remark]{rmk:taxonomy}
When $\psi$ is linear, \eqref{eq:researcher_problem} reduces to maximizing expected transformed welfare, or equivalently minimizing
$
R_\phi(A)
:=
V_\phi(G_\phi^\star)
-
\E_{P_{\tr}^n}[V_\phi\{A(S)\}].
$
If $\phi$ is concave, this objective is consistent with equality-minded or fairness-oriented treatment choice \citep{kitagawa_tetenov_2021,viviano_bradic_2024}. If $\phi$ is also linear, it is the standard expected-regret problem \citep{kitagawa_tetenov_2018}. When $\phi$ is linear and $\psi$ is concave, the PM is utilitarian but R is sampling-risk averse, which is the case studied below. Concavity of both transformations accommodates the two concerns simultaneously without requiring them to be linked. Appendix~\ref{app:joint_concavity} states the corresponding statistical results.
\end{remark}

\begin{remark}[Sampling risk averse treatment choice]\label[remark]{rmk:manski_tetenov}
Risk aversion over sampling is directly inspired by \citet{manski_tetenov_2007}. They study an innovation experiment without covariates, with binary outcomes, and with a known status quo success probability. A statistical treatment rule maps the experimental sample into the fraction of the population receiving the innovation, and the planner evaluates the rule through the sampling expectation of a concave transformation of the population mean outcome. It has the same preference structure as \eqref{eq:researcher_problem} with $\phi(y)=y$ and concave $\psi$. Their analysis studies admissibility of fractional treatment rules in that experiment. Ours connects sampling-risk preferences to algorithmic stability and welfare guarantees for learned policies with covariates. The mean-variance criterion introduced next is motivated by a local approximation to the exact expected-utility objective.
\end{remark}

The nested objective resembles several second-order models of uncertainty. Smooth ambiguity applies an outer transformation to inner expected utility;  variational and multiplier preferences penalize departures from a reference model;  and recursive models emphasize when uncertainty is resolved \citep{klibanoff_et_al_2005,maccheroni_marinacci_rustichini_2006,strzalecki_2011,strzalecki_2013,denti_pomatto_2022}. The robust mean-variance analysis of \citet{maccheroni_marinacci_ruffino_2013} is especially close mathematically. The distinction between objective and subjective mixtures in \citet{ghirardato_et_al_2003} helps interpret fractional policy randomization. Unlike those works, every probability law in our setting is objective, and the data-generating process is fixed.

\begin{remark}[Outcome levels versus policy gains]\label[remark]{rmk:gains}
The criterion $V_\phi(G)=\E\{\phi(Y^G)\}$ evaluates the distribution of post-policy outcome levels. Without additional structure, concavity of $\phi$ therefore represents aversion to dispersion in outcomes, not specifically dispersion in treatment gains. A gain-based alternative is
$
V_\phi^\Delta(G)
=
\E\left(\phi\left[D^G\{Y(1)-Y(0)\}\right]\right),
$
whose transformed conditional treatment contrast is
$
\tau_\phi^\Delta(x)
=
\E\left[\phi\{Y(1)-Y(0)\}-\phi(0)\mid X=x\right].
$
All compound-lottery arguments continue to apply after replacing $V_\phi$ by $V_\phi^\Delta$. In the utilitarian case $\phi(y)=y$ studied below, outcome welfare and gain welfare differ only by the policy-invariant constant $\E\{Y(0)\}$, so they induce the same policy target and regret.
\end{remark}

\subsection{Risk-aware regret}

For simplicity, take the policymaker to be utilitarian. Setting $\phi(y)=y$, population welfare is
\begin{equation*}\label{eq:welfare}
W(G)
:=
\E(Y^G)
=
\E\{Y(0)\}+\E\{\tau(X)G(X)\}.
\end{equation*}
Integrability is enough to make this welfare comparison finite. 
The unconstrained oracle is $G^\star(x)=\1\{\tau(x)\geq0\}$, with welfare $W^\star:=W(G^\star)$. Standard expected regret is $R(A):=W^\star-\E_{P_{\tr}^n}[W\{A(S)\}]$; when $\psi$ is linear, minimizing this quantity is exactly \eqref{eq:researcher_problem}.

A local expected-utility calculation motivates adding sampling variation to expected regret. Write $Z_A=W\{A(S)\}$ and define R's certainty equivalent by $\psi\{\CE_\psi(A)\}=\E\{\psi(Z_A)\}$. The next result makes the mean-variance approximation uniform.

\begin{lemma}[Mean-variance approximation]\label[lemma]{lem:taylor}
Let $\mathcal J$ be an open interval containing every feasible welfare value. Suppose $\psi\in C^3(\mathcal J)$, $\inf_{w\in\mathcal J}\psi'(w)>0$, and
$
-\frac{\psi''(w)}{\psi'(w)}=\gamma\geq0
$ for every $w\in\mathcal J.
$
There exists a finite constant $C_\psi$, depending only on $\psi$ and the feasible welfare interval, such that every $A\in\cA$ satisfies
$$
\CE_\psi(A)
=
\E(Z_A)-\frac{\gamma}{2}\V(Z_A)+r_\psi(A)
,\quad 
|r_\psi(A)|
\leq
C_\psi
\left\{
\E|Z_A-\E(Z_A)|^3+\V(Z_A)^2
\right\}.
$$
Consequently, the expansion is uniform over any sequence $\cA_n$ for which the centered third moments vanish uniformly:
$\sup_{A\in\cA_n}\E|Z_A-\E(Z_A)|^3\to0.
$
\end{lemma}

The lemma expresses the certainty equivalent as average welfare minus a common price for welfare variance. Constant absolute risk aversion makes that price the same across algorithms, while the smoothness condition controls the approximation uniformly.

\begin{corollary}[Optimizer approximation]\label[corollary]{cor:optimizer}
For a candidate class $\cA_n$, define the mean-variance objective
$
M_\rho(A)=\E(Z_A)-\rho\V(Z_A),
$ the approximation error $
\epsilon_n=\sup_{A\in\cA_n}|r_\psi(A)|,
$ and the risk aversion $
\rho=\gamma/2.
$
Consider the optimizers $A_n^{\CE}\in\argmax_{A\in\cA_n}\CE_\psi(A)$ and $A_n^{\mathrm{MV}}\in\argmax_{A\in\cA_n}M_\rho(A)$. Then
$
0
\leq
\CE_\psi(A_n^{\CE})-\CE_\psi(A_n^{\mathrm{MV}})
\leq
2\epsilon_n,
$
and
$
0
\leq
M_\rho(A_n^{\mathrm{MV}})-M_\rho(A_n^{\CE})
\leq
2\epsilon_n.
$ 
In particular, if the mean-variance maximizer is unique and its objective gap compared to the next best algorithm exceeds $2\epsilon_n$, then the certainty-equivalent and mean-variance maximizers coincide.
\end{corollary}

The corollary shows when the approximation can guide algorithm choice: a uniformly small approximation error implies a small loss in certainty equivalent from choosing an algorithm by its mean-variance approximation.

These calculations motivate our definition of risk-aware regret, which adds a penalty for sampling risk to ordinary expected regret. We study this criterion in its own right. The preceding results explain when it approximates the expected-utility choice.

\begin{definition}[Risk-aware regret]\label[definition]{def:riskaware}
For $\rho\geq0$, the risk-aware regret of algorithm $A$ is
\begin{equation}\label{eq:riskaware}
R_{\RA}(A)
:=
W^\star
-
\E[W\{A(S)\}]
+
\rho\V[W\{A(S)\}]
=
R(A)+\rho\V(Z_A).
\end{equation}
\end{definition}

The two terms in \eqref{eq:riskaware} separate expected performance from reliability. The first is ordinary expected welfare regret. The second prices variation in the welfare of the learned policy across experimental samples. The parameter $\rho$ has units inverse to welfare, so the terms are comparable. The criterion is a second-order certainty-equivalent loss.

\begin{remark}[Nonlinear regret]\label[remark]{rmk:nonlinear_regret}
Let $\operatorname{Reg}(A,S)=W^\star-W\{A(S)\}$. The mean-square-regret criterion of \citet{kitagawa_lee_qiu_2026} satisfies
$
\E\{\operatorname{Reg}(A,S)^2\}
=
R(A)^2+\V[W\{A(S)\}],
$
whereas \eqref{eq:riskaware} is $R(A)+\rho\V[W\{A(S)\}]$. The criteria share the same variance component but differ in the role of mean regret. Our criterion preserves the canonical expected-regret benchmark and allows the price of sampling risk to vary independently.
\end{remark}

\begin{remark}[Stochastic dominance]\label[remark]{rmk:sd}
\citet{manski_tetenov_2023} take the sampling distribution of welfare as the primitive object and require comparisons to respect stochastic dominance. Risk-aware regret also depends on the sampling distribution, but only through its initial two moments. A mean-variance functional need not respect first-order stochastic dominance globally because a stochastically better distribution can have a larger variance. Our criterion is therefore not a substitute for their general dominance framework. Its role is narrower: it is a locally expected-utility-based and statistically tractable target for studying how algorithmic stability controls sampling risk.
\end{remark}

Risk-aware regret is a local mean-variance criterion. It augments standard expected regret without claiming to order arbitrary welfare distributions. Using its two components, we interpret algorithm choice as a tradeoff between expected welfare and sampling risk.
\section{A tradeoff: Expected welfare versus sampling risk}\label{sec:design}

The researcher manages sampling risk by choosing among learning algorithms before seeing the data. The choice is the same $A\in\cA$ introduced in \Cref{sec:compound}, and the recommendation is $\widehat G=A(S)$. Different tuning parameters are incorporated as different elements of $\cA$.\footnote{We write $\cA_n$ only when sample-size dependence matters. For complete policy-vote bagging, the candidates will be $\cA=\{A_m:1\leq m\leq n\}$.}

\begin{example}[Deterministic rule]\label[example]{ex:plugin}
Hard thresholding can make a deterministic policy sensitive near zero. Let $\widehat\tau_h$ be a CATE estimator whose hyperparameter $h$ may index bandwidth, tree depth, leaf size, or regularization strength. A deterministic algorithm $A\in\cA$ yields  the policy
$
\widehat G(x)
=
\{A(S)\}(x)
=
\1\{\widehat\tau_h(S,x)\geq0\}.
$
Here $\widehat\tau_h$ is the score function and $t\mapsto\1(t\geq0)$ is the choice function. The hyperparameter $h$ is absorbed into the definition of the algorithm $A$. Because the choice function is discontinuous, small changes in the training sample can flip treatment decisions near zero and thereby create sampling variability in welfare.
\end{example}

\begin{example}[Fractional rule]\label[example]{ex:fractional}
A fractional candidate smooths the treatment choice near zero:
$
\widehat G(x)
=
\{A(S)\}(x)
=
\delta_c\{\widehat\tau_h(S,x)\},
$
where $\delta_c(t)=0$ if $t\leq -c$, $\delta_c=(t+c)/(2c)$ if $|t|\leq c$, and $\delta_c(t)=1$ if $t\geq c$ for some $c>0$.
The pair $(h,c)$ is again absorbed into the definition of $A$. Increasing $c$ reduces sensitivity near the treatment threshold, potentially lowering sampling risk at the cost of expected welfare.
\end{example}

A risk-aware research design chooses the best attainable combination of expected regret and welfare variation. For the candidate class $\cA$, R solves
$
A_\star
\in
\argmin_{A\in\cA}
R_{\RA}(A).
$ 
The lower-left Pareto boundary summarizes the attainable tradeoff. Each algorithm generates a pair $\bigl(R(A),\V(Z_A)\bigr)$, and no algorithm on the boundary can be improved in both expected regret and welfare variance by another candidate. As in the mean-variance frontier of portfolio theory, different policy learning algorithms offer different combinations of return and risk \citep{markowitz1952}.  
Here the frontier concerns population welfare across training samples, and $\rho$ determines how much expected welfare R is willing to exchange for greater reliability.

The frontier gives a geometric interpretation of the optimal algorithm. For each $\rho>0$, R selects the attainable point on the lowest iso-risk-aware-regret line. At a smooth interior solution, that line is tangent to the lower-left boundary, and a larger $\rho$ favors a point with less sampling risk. \Cref{fig:frontier} illustrates this geometry. Appendix~\ref{app:sec3_proofs} records the slope calculation.

\begin{figure}[t]
\centering
\begin{tikzpicture}[
    x=1cm, y=0.9cm,
    font=\normalsize,
    >=Latex,
    axis/.style={->,line width=0.9pt,black!75},
    frontier/.style={RoyalBlue!85!black,line width=2pt},
    iso/.style={
        black!65,line width=1.2pt,
        dash pattern=on 4pt off 3pt
    }
]

% Axes
\draw[axis] (0,0) -- (6.1,0);
\draw[axis] (0,0) -- (0,3.9);
\node[below=7pt] at (3.05,0)
    {Expected regret $R(A)$};
\node[rotate=90,above=7pt] at (0,1.95)
    {Welfare variance $\V(Z_A)$};

% Selected point and exact tangent slope
\pgfmathsetmacro{\rstar}{3}
\pgfmathsetmacro{\vstar}{4.8/(\rstar+0.45)-0.1}
\pgfmathsetmacro{\tangentslope}{-4.8/(\rstar+0.45)^2}

% Iso-risk-aware-regret line
\draw[iso,domain=0.8:5.9,samples=2,variable=\x]
    plot ({\x},{\vstar+\tangentslope*(\x-\rstar)});

% Efficient frontier
\draw[frontier,domain=0.9:5.85,smooth,
      samples=100,variable=\x]
    plot ({\x},{4.8/(\x+0.45)-0.1});

% Optimal choice
\fill (\rstar,\vstar) circle[radius=2.2pt];
\node[above right=3pt] at (\rstar,\vstar) {$A_\star$};

\end{tikzpicture}

\caption{Algorithmic design and the welfare-risk frontier}
\label{fig:frontier}
\caption*{\footnotesize \textit{Notes:}
The blue curve shows the efficient combinations of expected
regret and welfare variance attainable within $\cA$.
Along the dashed line, $R(A)+\rho\V(Z_A)$ is constant.
At the illustrated optimum $A_\star$, this line is tangent
to the frontier, with slope $-1/\rho$.}
\end{figure}

Estimating the frontier of expected regret and welfare variation is a separate inferential problem. The support-function methods of \citet{liu_molinari_2025} for a fairness-accuracy frontier may provide a useful starting point, but we do not estimate the boundary here. We instead ask which aspects of an algorithm control its position on the welfare variation axis. Algorithmic stability supplies the answer.
\section{The central role of algorithmic stability}\label{sec:stability}

This section shows how a property of the policy learning algorithm controls sampling risk. We define welfare stability, then use it to bound welfare variance and downside risk. Finally, we connect those guarantees to expected utility under constant absolute risk aversion.

\subsection{Welfare stability}

Welfare stability asks how much one observation can change the value of the learned policy. Fix $A:\cS^n\to\cG$ and write $A(S)=\widehat G_S$. For $i\in\{1,\ldots,n\}$ and $s'\in\cS$, let $S^{i\leftarrow s'}=(S_1,\ldots,S_{i-1},s',S_{i+1},\ldots,S_n)$ denote the sample obtained by replacing observation $i$.

\begin{definition}[Welfare stability]\label[definition]{def:stability}
The welfare-stability parameter of $A$ is
\begin{equation*}\label{eq:stability}
\beta_n(A)
:=
\sup_{S\in\cS^n}
\sup_{1\leq i\leq n}
\sup_{s'\in\cS}
\left|
W\{A(S)\}
-
W\{A(S^{i\leftarrow s'})\}
\right|.
\end{equation*}
A sequence $A_n$ is \emph{welfare stable} if $\beta_n(A_n)\to0$. 
\end{definition}

Welfare stability is weaker than requiring every prediction to be stable. It controls the scalar welfare of the policy, so changes at different covariate values may offset one another. Welfare stability requires only that the influence of one observation vanishes. For the generic variance bound below to vanish, the stronger rate $\beta_n(A_n)=o(n^{-1/2})$ is needed. The next result gives simple conditions under which familiar stability guarantees imply welfare stability.

\begin{lemma}[From algorithmic stability to welfare stability]\label[lemma]{lem:policy_to_welfare}
Suppose Assumption \ref{ass:main}(i) holds and define the algorithmic stability parameter 
\[
\Gamma_n(A)
=
\sup_{S,i,s'}\sup_{x\in\cX}
\left|
\{A(S)\}(x)-\{A(S^{i\leftarrow s'})\}(x)
\right|.
\]
Then the welfare stability parameter satisfies
$
\beta_n(A)
\leq
\Gamma_n(A) \lVert\tau\rVert_{L_1(P_X)}.
$ 

In particular, if $\{A(S)\}(x)=\delta\{\zeta(S,x)\}$, the choice function $\delta$ is $L_\delta$-Lipschitz, and the score function $\zeta$ has algorithmic stability parameter at most $s_n$, then
$
\beta_n(A)
\leq
L_\delta s_n\lVert\tau\rVert_{L_1(P_X)}.
$
\end{lemma}

The lemma makes welfare stability verifiable through well-known algorithmic stability results.\footnote{Here, $\Gamma_n$ is standard, while $s_n$ follows from e.g. \citet{sun_xiang_2025}.} 
The next subsection turns this welfare sensitivity into variance and risk guarantees.

\subsection{Main result}

Our central result translates welfare stability into direct control of sampling risk.

\begin{theorem}[Stable policy learning]\label[theorem]{thm:stability}
If $S_1,\ldots,S_n$ are independent and $\beta_n(A)<\infty$, then
\begin{equation*}\label{eq:stable_ra}
R_{\RA}(A)
\leq
R(A)
+
\rho \frac{n\beta_n(A)^2}{4}.
\end{equation*}
Moreover, for every $t>0$,
\[
\Pp\left[
W^\star-W\{A(S)\}
\geq
R(A)+t
\right]
\leq
\exp\left\{-\frac{2t^2}{n\beta_n(A)^2}\right\}.
\]
\end{theorem}

The theorem translates algorithmic stability into two policy learning guarantees. First, it bounds the variance term in risk-aware regret. Second, it bounds the probability that the learned policy performs poorly. The welfare stability parameter must vanish faster than $n^{-1/2}$ for this worst-case variance bound to disappear.

\subsection{Utility interpretation}

Constant absolute risk aversion gives an exact link between stability and expected utility. For $\gamma>0$, define the CARA utility function $\psi_\gamma(w)=-\exp(-\gamma w)$ and the CARA certainty equivalent $\CE_\gamma(A)=-(1/\gamma)\log[\E\{\exp(-\gamma Z_A)\}]$.\footnote{By definition, the certainty equivalent solves $\psi_\gamma\{\CE_\gamma(A)\}
=\E\{\psi_\gamma(Z_A)\}$.}

\begin{proposition}[Exact CARA guarantee from stability]\label[proposition]{prop:cara_stability}
Suppose $Z_A=W\{A(S)\}$ is integrable and $A$ has welfare-stability parameter $\beta_n(A)<\infty$. Then
$
W^\star-\CE_\gamma(A)
\leq
R(A)
+
\frac{\gamma n}{8}\beta_n(A)^2.
$
\end{proposition}

The CARA guarantee gives a direct utility interpretation to welfare stability. It applies to any learner with algorithmic stability and requires no approximation. The next result instead uses welfare variance itself, which can be much smaller than its worst-case bound.

\begin{proposition}[Exact CARA guarantee from welfare variance]\label[proposition]{prop:cara_variance}
Suppose Assumption~\ref{ass:main}(i) holds, $\gamma>0$, and $T=\lVert\tau\rVert_{L_1(P_X)}>0$. Then every policy-learning algorithm satisfies
\[
R(A)
\leq
W^\star-\CE_\gamma(A)
\leq
R(A)+
\frac{e^{\gamma T}-1-\gamma T}{\gamma T^2}\V(Z_A).
\]
If $T=0$, expected regret, welfare variance, and certainty-equivalent loss are all zero.
\end{proposition}

The variance coefficient approaches $\gamma/2$ as $\gamma T\to0$, matching the local mean-variance calculation. The inequality itself holds for every $\gamma>0$ and every sample size. Any sharper bound on welfare variance can therefore also sharpen the exact utility guarantee; policy-vote bagging will provide such a bound.

Under the smooth preference assumptions of \Cref{lem:taylor}, stability also controls the error in the local mean-variance approximation. The following corollary gives a sufficient rate for that error to be negligible uniformly over the policy learning algorithms.

\begin{corollary}[Stability validates the approximation]\label[corollary]{cor:stable_taylor}
Suppose the conditions of \Cref{lem:taylor} hold for a class $\cA_n$, and $\sup_{A\in\cA_n}\beta_n(A)\leq\bar\beta_n$. Then the optimizer-level approximation error in \Cref{cor:optimizer} satisfies
$
\epsilon_n
\lesssim
(n\bar\beta_n^2)^{3/2}
+
(n\bar\beta_n^2)^{2}.
$
Consequently, if $n\bar\beta_n^2\to0$, the certainty-equivalent and mean-variance objectives are uniformly asymptotically equivalent, and the remainder is of smaller order than the generic variance scale $n\bar\beta_n^2$.
\end{corollary}

Generality is also the limitation of the stability theorem. Because it uses only worst-case analysis, it can be conservative for structured algorithms. Policy-vote bagging has additional symmetry that yields a sharper variance guarantee.
\section{A concrete proposal}\label{sec:bagging}

Policy-vote bagging creates a stable final policy by averaging decisions learned on overlapping subsamples. Each base learner casts a binary treatment vote, and the fraction voting for treatment becomes the final treatment probability. The base learner may be unstable on its own $m$-unit subample. The aggregate can nevertheless be stable because replacing one unit leaves every vote from a subsample that omits it unchanged. Under complete bagging, only the $m/n$ fraction of subsamples containing that unit can respond. The procedure thus builds stability across subsamples into stability across the full experimental sample. The inclusion-frequency argument appears in the bagging stability analysis of \citet{elisseeff_evgeniou_pontil_2005} and is used by \citet{chernozhukov_newey_singh_syrgkanis_2026} and \citet{chen_syrgkanis_austern_2022} to stabilize nuisance estimators for debiased inference. Here it controls the welfare of a learned treatment policy.

\subsection{Policy-vote bagging}

Complete policy-vote bagging averages over every subsample of a fixed size. For $1\leq m\leq n$, write $\cI_{n,m}=\{I\subseteq\{1,\ldots,n\}:|I|=m\}$ and $S_I=(S_i:i\in I)$. A base learner maps $S_I$ into a CATE estimate $\widehat\tau_I$. We assume the learner is symmetric: permuting the observations within the subsample leaves the fitted function $\widehat{\tau}_I$ unchanged. This estimate casts the vote $\widehat G_I(x)=\1\{\widehat\tau_I(x)\geq0\}$. The complete policy is
\begin{equation}\label{eq:bagged_policy}
\{A_m(S)\}(x)
=:
\widehat G_m^{\mathrm{comp}}(x)
=
\binom{n}{m}^{-1}
\sum_{I\in\cI_{n,m}}\widehat G_I(x)
=
\E_I\{\widehat G_I(x)\mid S\}.
\end{equation}
Retaining vote shares as treatment probabilities is essential to the method. Policy-vote bagging averages the decisions $\1\{\widehat\tau_I(x)\geq0\}$, whereas score bagging averages numerical CATE estimates before applying a threshold. Classical bagged classification also aggregates votes, but typically takes a final majority decision \citep{breiman_1996}. We retain the fractional average, since applying a final threshold can destroy its stability  \citep{soloff_barber_willett_argmax_2024}. The candidate class here is $\cA=\{A_m:1\leq m\leq n\}$: a smaller $m$ makes each unit influential in fewer votes, while a larger $m$ gives each base learner more information.

Computational policy-vote bagging approximates the complete average with a bounded number of subsamples. Let $I_1,\ldots,I_B$ be conditionally i.i.d. uniform draws from $\cI_{n,m}$, independent of $S$. Define the realized bag collection $\mathbf I=(I_1,\ldots,I_B)$ and 
\begin{equation}\label{eq:finite_bagging}
\{A_{m,B}(S,\mathbf I)\}(x)
=:
\widehat G_{m,B}(x)
=
\frac{1}{B}\sum_{b=1}^B\widehat G_{I_b}(x).
\end{equation}
Computational bagging introduces computational randomness in addition to experimental-sampling randomness. Risk-aware regret for $A_{m,B}$ averages over both sources. The displayed definitions use a deterministic base learner. For a randomized base learner, complete bagging also averages over its seed, and computational bagging draws an independent seed with each bag; this seed variation can persist even at $m=n$. Conditional on $S$, the policy in \eqref{eq:bagged_policy} is the expectation of the policy in \eqref{eq:finite_bagging}, so complete bagging removes Monte Carlo variation from aggregation. Classical random forests are computational bagging procedures with randomized tree learners \citep{breiman_2001}.

Averaging policy votes has a direct expected-utility benefit at any fixed bag size. Because welfare is affine in treatment probabilities, averaging preserves mean welfare while reducing the variation created by selecting a particular collection of bags. The next result makes this comparison for any concave evaluation of welfare.

\begin{proposition}[Averaging votes improves concave welfare evaluations]\label[proposition]{prop:bag_concave}
Suppose Assumption~\ref{ass:main}(i) holds. Let $\psi$ be increasing and concave on an open interval containing the feasible welfare interval. For every $m$ and $B\geq1$, conditional on the experimental sample,
\begin{align*}
\psi[W\{A_m(S)\}]
&\geq
\E_{\mathbf I}\left(\psi\left[W\{A_{m,B}(S,\mathbf I)\}\right]\mid S\right)
\geq
\E_I[\psi\{W(\widehat G_I)\}\mid S].
\end{align*}
All three policies have the same conditional mean welfare:
\[
W\{A_m(S)\}
=
\E_{\mathbf I}[W\{A_{m,B}(S,\mathbf I)\}\mid S]
=
\E_I\{W(\widehat G_I)\mid S\}.
\]
The same utility ordering and mean equality hold after averaging over the experimental sample.
\end{proposition}

Complete bagging is therefore weakly preferred to computational bagging, which is weakly preferred to deploying the policy from one randomly selected bag, under any increasing concave welfare evaluation. This improvement holds at a fixed $m$ and preserves expected welfare. Choosing the bag size raises the additional question of how much information each base learner should receive, to which we turn in Section~\ref{sec:rates}.

Policy-vote bagging acts like a data-adaptive smoother around the treatment threshold. At a covariate value $x$, each bag supplies a score $\widehat\tau_{I_b}(x)$, and the final policy averages the hard-threshold decision over the empirical distribution of those scores. That distribution behaves like a smoothing kernel whose shape and width are generated by subsampling rather than chosen in advance. Smaller bags can produce more dispersed scores and a wider fractional-treatment region, while larger bags can sharpen the decision boundary.\footnote{These are possible effects of bag size, rather than a monotonicity property required by our results.} 
The number of bags $B$ controls the Monte Carlo resolution. This score-space view complements the  covariate-space smoothing induced by forest weights \citep{wager_athey_2018}.

\subsection{Policy-vote bagging is stable}\label{sec:bag_stability}

A key benefit of policy-vote bagging is welfare stability. The next lemmas verify stability for both complete and computational variants of policy-vote bagging. 

\begin{lemma}[Welfare stability of complete policy-vote bagging]\label[lemma]{lem:bag_stability}
Under Assumption \ref{ass:main}(i), the complete policy-vote bagging algorithm in \eqref{eq:bagged_policy} satisfies
$
\beta_n(A_m)
\leq
\frac{m}{n}\lVert\tau\rVert_{L_1(P_X)}.
$
The factor $m/n$ is worst-case sharp over symmetric binary base learners and data-generating processes satisfying Assumption \ref{ass:main}.
\end{lemma}

\begin{lemma}[Welfare stability of computational policy-vote bagging]\label[lemma]{lem:finite_stability}
Under Assumption \ref{ass:main}(i), the computational policy-vote bagging algorithm in \eqref{eq:finite_bagging} satisfies,
with probability at least $1-\eta$ over the computational bags,
$
\beta_n(A_{m,B}\mid\mathbf I)
\leq
\lVert\tau\rVert_{L_1(P_X)}
\left\{
\frac{m}{n}
+
\sqrt{\frac{\log(n/\eta)}{2B}}
\right\}.
$
\end{lemma}

The stability of bagging is governed by how often each unit appears. Complete bagging gives every unit the exact inclusion frequency $m/n$, whereas computational bagging uses its realized frequency among the $B$ sampled bags. These welfare guarantees complement the assumption-free prediction-stability results of \citet{soloff_barber_willett_2024}. The symmetry of complete bagging yields a sharper variance rate than worst-case sensitivity alone.

\subsection{Upper bounds for policy-vote bagging}\label{sec:bag_upper_bounds}

Stability is useful only if it can be achieved without sacrificing too much welfare. The next theorem expresses expected regret and welfare variance in terms of the bag-level CATE error, making the accuracy-stability tradeoff explicit.

Formally, the accuracy of the base learner enters through its integrated CATE error. For an independent target covariate $X$, define
\begin{equation*}\label{eq:Qmp}
\Delta_I(X)
=
\widehat\tau_I(X)-\tau(X),
\quad
Q_{m,p}
=
\E_{S,I,X}\{|\Delta_I(X)|^p\},
\quad p>1.
\end{equation*}
Here, $Q_{m,2}$ is the integrated mean-squared CATE error of a learner trained on one size-$m$ bag. We initially require only that it be finite. \Cref{sec:learning_examples} shows how standard learning rates bound this quantity in parametric, H\"older-smooth, and Gaussian-kernel examples.

\begin{theorem}[Upper bound for complete policy-vote bagging]\label[theorem]{thm:bag_upper}
Suppose Assumption \ref{ass:main} holds, the base learner is symmetric, and $Q_{m,2}<\infty$. Then
\begin{equation*}\label{eq:main_upper}
R_{\RA}(A_m)
\lesssim
Q_{m,2}^{\frac{\kappa+1}{\kappa+2}}
+
\rho\frac{m}{n}Q_{m,2}.
\end{equation*}
\end{theorem}

The theorem separates accuracy from stability. The first bound converts CATE error into expected welfare regret: mistakes matter most when treatment effects are large, while the margin condition limits how often small estimation errors can reverse the correct decision. The second bound shows that complete averaging reduces welfare variance by the inclusion frequency $m/n$. 
Only the regret bound uses the margin condition. 

\begin{corollary}[Upper bound for computational policy-vote bagging]\label[corollary]{cor:finite_bagging}
Under the conditions of \Cref{thm:bag_upper}, for every $B\geq1$,
\[
R_{\RA}(A_{m,B})
\lesssim
Q_{m,2}^{\frac{\kappa+1}{\kappa+2}}
+
\rho\left(\frac{m}{n}+\frac{1-m/n}{B}\right)Q_{m,2}.
\]
\end{corollary}

The computational-bag variance bound interpolates between $Q_{m,2}$ at $B=1$ and $(m/n)Q_{m,2}$ as $B$ grows. For the deterministic base learner defined above, the Monte Carlo component vanishes as $B\to\infty$ and is zero at $m=n$, when there is only one subsample. Without a cost of computation, the bound weakly favors more bags; an interior choice of $B$ would require an explicit computation constraint.

The same variance bounds control exact CARA utility loss. Combining them with \Cref{prop:cara_variance} gives a direct preference-based guarantee for both complete and computational policy-vote bagging.

\begin{corollary}[Exact CARA guarantee for policy-vote bagging]\label[corollary]{cor:cara_bagging}
Under the conditions of \Cref{thm:bag_upper}, let $\gamma>0$ and $T=\lVert\tau\rVert_{L_1(P_X)}>0$. Then
\[
W^\star-\CE_\gamma(A_m)
\leq
R(A_m)+
\frac{e^{\gamma T}-1-\gamma T}{\gamma T^2}
\frac{m}{n}Q_{m,2}.
\]
For computational bagging, the same inequality holds with $A_m$ replaced by $A_{m,B}$ and $m/n$ replaced by $m/n+(1-m/n)/B$, with the certainty equivalent taken over both sampling and computational randomness. When $T=0$, both certainty-equivalent losses are zero.
\end{corollary}

For fixed $\gamma$ and $T$, exact CARA utility loss thus has the same upper-bound rate as risk-aware regret, including for the full-sample policy $m=n$.

\begin{remark}[Other CATE moments]\label[remark]{rmk:p_moment}
For every $p>1$, the regret conclusion has the corresponding moment form
$
R(A_m)
\lesssim
Q_{m,p}^{\frac{\kappa+1}{\kappa+p}}.
$
When the available primitive is an $L_2$ CATE rate, setting $p=2$ directly is sharper than first converting the MSE into a different moment. This is the route used in \Cref{thm:bag_upper}.
\end{remark}

The upper-bound analysis therefore reduces policy-vote bagging to two quantities: the base learner's CATE error and each unit's bag-inclusion frequency. The next section asks whether this decomposition can be improved and what bag size it recommends when CATE accuracy follows a standard learning rate.

\section{Lower bounds and optimality}\label{sec:rates}

Are the bounds sharp? What bag sizes do they recommend? How do the resulting rates compare with what any algorithm can achieve? We answer these questions in turn. First, we establish matching lower bounds over a high-level class of populations and base learners. Second, we characterize rate-optimal bag sizes. Third, we give a minimax benchmark over all policy-learning algorithms and use it to assess the parametric and Gaussian-kernel rates.

\subsection{Sharpness over a high-level class}\label{sec:bag_lower}

A single learning-rate parameter summarizes how the base learner improves with bag size. Suppose that, for some $\alpha \in(0,1]$, $
Q_{m,2}
\lesssim
m^{-\alpha}
$ uniformly over $1\leq m\leq n.
$
For empirical-risk minimizers, $Q_{m,2}$ is often the square of a critical radius, so localized complexity calculations can characterize $\alpha$ for particular CATE learners \citep{van_der_laan_2026}. Parametric learners have $\alpha=1$. In nonparametric problems, $\alpha<1$ reflects smoothness and dimension.

The learning rate $\alpha$ determines how bag size $m$ affects the two terms of our upper bound. Substituting the learning rate into \Cref{thm:bag_upper} gives
\begin{equation}\label{eq:rate_upper}
R_{\RA}(A_m)\lesssim m^{-\frac{\alpha(\kappa+1)}{\kappa+2}}+\rho\frac{m^{1-\alpha}}{n}.
\end{equation}
The first term bounds expected regret and decreases with bag size. The second term bounds the sampling-risk. Larger bags improve estimation accuracy, but each observation also enters a larger fraction of the bags. For $\rho>0$ and $\alpha<1$, accuracy improves too slowly to offset this increased influence, so the second term increases with $m$. When $\alpha=1$, the two effects exactly offset, leaving the second term constant.

Computational bagging adds variation from sampling the bags. With $B$ bags, \Cref{cor:finite_bagging} adds $\rho\{(1-m/n)/B\}m^{-\alpha}$ to the complete-bagging bound in \eqref{eq:rate_upper}. This additional term vanishes at $m=n$, where every bag is the full sample.

Our lower bound shows that neither term in \eqref{eq:rate_upper} can be improved without restricting the data-generating process or the base learner further.

\begin{theorem}[Lower bound for complete policy-vote bagging]\label[theorem]{thm:sharpness}
Fix $\kappa>0$ and $\alpha>0$. There exist a data-generating process satisfying Assumption \ref{ass:main},
and a sequence of symmetric bag-level base learners, such that for every $n\geq m\geq1$, we have
$
Q_{m,2}
\asymp
m^{-\alpha}
$ while 
\begin{equation}\label{eq:sharp_lower}
R_{\RA}(A_m)
\gtrsim
m^{-\frac{\alpha(\kappa+1)}{\kappa+2}}
+
\rho\frac{m^{1-\alpha}}{n}.
\end{equation}
The implicit constants are independent of $n$, $m$, and $\rho$. Together, \eqref{eq:rate_upper} and \eqref{eq:sharp_lower} show that the risk-aware-regret bound is sharp up to constants over the stated high-level class.
\end{theorem}

The theorem establishes sharpness when only an integrated CATE-error rate is known. It constructs a population and a sequence of base learners for which both terms attain their upper-bound orders. Particular learners can perform better. The matching bounds nevertheless allow a worst-case comparison of bag sizes over a class with common constants.

\subsection{Optimal bag size}\label{sec:bag_choice}\label{sec:learning_examples}

The matching bounds tell us how to choose bag size when only the learning rate is known. The comparison uses the same class of populations and base learners at every bag size.

\begin{corollary}[Optimal bag size over the high-level class]\label[corollary]{cor:bag_choice}
Fix $\kappa>0$ and $\alpha>0$. Let $\mathcal H_{\alpha,\kappa}$ satisfy the uniform class conditions in Appendix~\ref{app:bag_class}. Write $R_{\RA,H}(A_m)$ for the risk-aware regret of complete policy-vote bagging under the population and base learners specified by $H$. Uniformly over $n\geq1$ and $\rho\geq0$,
\[
\inf_{1\leq m\leq n}
\sup_{H\in\mathcal H_{\alpha,\kappa}}
R_{\RA,H}(A_m)
\asymp
\inf_{1\leq m\leq n}
\left\{
m^{-\frac{\alpha(\kappa+1)}{\kappa+2}}
+\rho\frac{m^{1-\alpha}}{n}
\right\},
\]
where the comparison constants depend only on $\alpha,\kappa$ and the class bounds. Consequently:
\begin{enumerate}[label=(\roman*),leftmargin=2.1em]
\item If $\alpha=1$, choosing $m=n$ is minimax-rate optimal.

\item If $0<\alpha<1$ and $\rho>0$, let $m_0=
\left(\frac n\rho\right)^{
\frac{\kappa+2}{\kappa+2-\alpha}}.$ Any integer bag size with $m\asymp\min\{n,\max(1,m_0)\}$ is minimax-rate optimal.

\item For every fixed $\rho\geq0$ and every $\alpha>0$, the minimized worst-case criterion has order $n^{-\alpha(\kappa+1)/(\kappa+2)}$ as $n\to\infty$. Every sequence $m\asymp n$, including $m=n$, attains this order.
\end{enumerate}
\end{corollary}

\Cref{cor:bag_choice} clarifies when to use a bag size $m$ that is smaller than the sample size $n$. 
If estimation accuracy improves quickly due to a parametric CATE $(\alpha=1)$, the researcher should use the full sample, taking $m=n$. 
If estimation accuracy improves slowly due to a nonparametric CATE $(\alpha<1)$, the researcher should typically use subsamples, taking $m=\left(\frac n\rho\right)^{
\frac{\kappa+2}{\kappa+2-\alpha}}$. 
Intuitively, if there is simple heterogeneity with a known functional form ($\alpha=1$), then using the full sample is fine; if there is complex heterogeneity with an unknown functional form ($\alpha<1$), then subsampling helps.
Moreover, placing a higher price $\rho$ on sampling risk favors smaller bags. 
For a fixed and large price $\rho$ of sampling risk, the researcher should use $m<n$ in finite samples and $m=n$ in the large sample limit $n\rightarrow \infty$. \Cref{fig:bag_size_envelope} illustrates these comparisons. 

\begin{figure}[t]
\centering
\includegraphics[width=\textwidth]{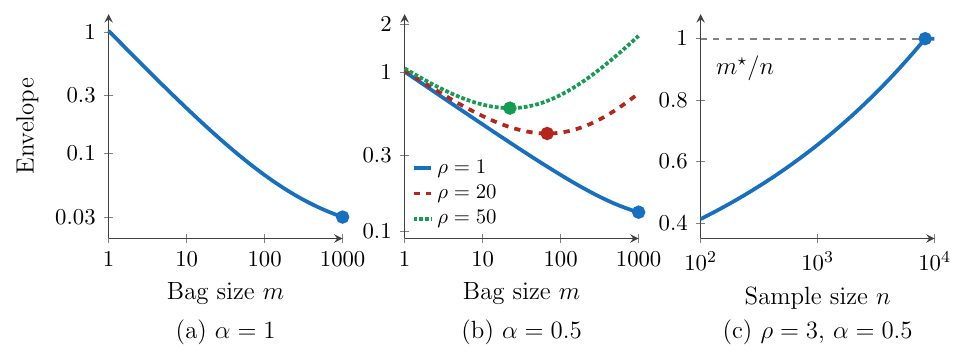}
\caption{Bag size and the risk-aware-regret bound}
\label{fig:bag_size_envelope}
\caption*{\footnotesize \textit{Notes:}
All panels use the envelope in \eqref{eq:rate_upper}
with unit leading constants and $\kappa=1$.
Panels (a) and (b) fix $n=1000$; panel (a) uses $\rho=20$.
Dots in these panels mark the minimizing bag sizes,
treating $m$ as continuous and imposing $1\leq m\leq n$.
Panel (b) shows how greater weight on sampling risk
can favor smaller bags.
Panel (c) fixes $\rho=3$ and plots the minimizing
fraction $m^\star/n$ as the sample grows.
The dashed line marks the full-sample choice,
which eventually minimizes the envelope.}
\end{figure}

The next two examples show how the learning rate affects optimal bag size.

\begin{example}[Parametric CATE]\label[example]{ex:parametric}
Under a parametric model and correctly-specified learner, $Q_{m,2}\lesssim m^{-1}$ and
$
R_{\RA}(A_m)
\lesssim
m^{-(\kappa+1)/(\kappa+2)}
+
\frac{\rho}{n}.
$
The rate-optimal bag size is $m=n$.
\end{example}

Slower learning leaves more room for smaller bags to help. The next example makes this dependence explicit in terms of smoothness and dimension.

\begin{example}[H\"older-smooth CATE]\label[example]{cor:holder}
Under the H\"older model and local-polynomial learner specified in Appendix~\ref{app:holder_conditions}, let $s$ denote smoothness and $d$ the covariate dimension. Then
$
Q_{m,2}\lesssim m^{-\frac{2s}{2s+d}}
$ and $
R_{\RA}(A_m)\lesssim
m^{-\frac{2s(\kappa+1)}{(2s+d)(\kappa+2)}}
+\rho\frac{m^{\frac{d}{2s+d}}}{n}.
$
The rate-optimal bag size is $m=\left(\frac n\rho\right)^{
\frac{(2s+d)(\kappa+2)}
{2s(\kappa+1)+d(\kappa+2)}}$.
\end{example}

This bound need not be sharp for smooth treatment effects. More detailed guarantees for the estimator can give faster regret rates \citep{audibert_tsybakov_2007}.

\begin{remark}[Finite-sample bag choice]\label[remark]{rem}\label[remark]{rem:finite_sample_bag_choice}
A rate-optimal bag size need not minimize finite-sample risk-aware regret. At fixed $n$, the upper bound depends on the actual error curve $m\mapsto Q_{m,2}$ and the constants omitted by rate statements. Finding the best bag size requires comparing expected welfare and sampling variance. Estimating the frontier is for future work.
\end{remark}

CARA utility permits an exact comparison in simple experiments. Even with parametric estimation accuracy, sufficiently strong risk aversion can favor smaller bags.

\begin{example}[Finite-sample bag choice in a toy example]\label[example]{ex:cara_bag_choice}
Consider the experiment and base learner in Appendix~\ref{app:cara_experiment}, with $n=3$ and $Q_{m,2}=75/(64m)$. Under $\psi_\gamma(w)=-\exp(-\gamma w)$, complete policy-vote bagging satisfies $\argmax_{m\in\{1,2,3\}}\CE_\gamma(A_m)=3$ if $0<\gamma<4\log 3$, and $\argmax_{m\in\{1,2,3\}}\CE_\gamma(A_m)=1$ if $\gamma>4\log 3$.
\end{example}

The full sample has higher expected welfare, but strong enough risk aversion favors small bags. We conclude that the best finite-sample choice depends on more than the learning rate.

\subsection{Minimax benchmark over all algorithms}\label{sec:minimax_local}

The final question is what any policy-learning algorithm can achieve. We partially answer the question, focusing on the parametric case $(\alpha=1)$. The next theorem uses two populations that differ only in whether treatment helps a small group.

\begin{theorem}[Local minimax rate]\label[theorem]{thm:minimax}
Fix $\kappa>0$, and let $\mathcal P_n$ be the parametric experiments constructed in Appendix~\ref{app:local_experiment}. For sufficiently small $c>0$ and every $\rho\geq0$,
\begin{equation*}\label{eq:minimax}
\inf_{A:\cS^n\to\cG}
\sup_{P\in\mathcal P_n}
R_{\RA,P}(A)
\asymp
n^{-(\kappa+1)/(\kappa+2)}.
\end{equation*}
The constants are independent of $n$ and $\rho$.
\end{theorem}

The lower bound comes from mistakes in treatment choice. It matches the parametric rate in \Cref{ex:parametric} for fixed $\rho$.

Gaussian kernels offer a flexible learning method with nearly the same rate. The next example uses \Cref{thm:minimax} to show when this rate is close to the best possible.

\begin{example}[Gaussian-RKHS CATE]\label[example]{cor:gaussian_rkhs}
Under the Gaussian-RKHS model and kernel ridge learner specified in Appendix~\ref{app:gaussian_conditions}, let $A_n$ be the full-sample policy and $d$ the covariate dimension. Uniformly over the stated class,
$
Q_{n,2}\lesssim\frac{(\log n)^d}{n}
$ and $
R_{\RA}(A_n)\lesssim
\left\{\frac{(\log n)^d}{n}\right\}^{\frac{\kappa+1}{\kappa+2}}
+\rho\frac{(\log n)^d}{n}.
$
Under the additional lower-bound conditions stated there, this rate is minimax for fixed $\rho$, up to the factor $(\log n)^{d(\kappa+1)/(\kappa+2)}$.
\end{example}

The Gaussian comparison allows the covariate distribution to vary within the class. It does not establish the same lower bound for every fixed distribution.
\section{Discussion}\label{sec:discussion}

This paper connects sampling risk in statistical treatment choice to algorithmic stability. 
The compound lottery includes both variation across experimental samples and variation across the full population under the implemented policy. 
The main text focuses on aversion to the first source of uncertainty, while Appendix~\ref{app:joint_concavity} allows risk aversion at both stages. 
If replacing one experimental unit changes a policy recommendation only slightly, then welfare cannot vary much across experimental samples. 
Policy-vote bagging builds this stability into the overall policy learning algorithm.

Future work may estimate the frontier of expected welfare and sampling risk to guide the choice of bag size. Accounting for computational costs could also guide how many bags to use. Other extensions include alternative risk measures and treatment capacity constraints.

\begingroup
\raggedright
\setlength{\bibsep}{4pt plus 1pt minus 1pt}

\endgroup

\clearpage
\appendix
\begin{center}
{\Large\bfseries Appendix}
\end{center}
\vspace{1.5em}
\appcontentsline{app:joint_concavity}{Appendix A. Extension to joint risk aversion}
\appcontentsline{app:sec2_proofs}{Appendix B. Proofs for Section 2}
\appcontentsline{app:sec3_proofs}{Appendix C. Proofs for Section 3}
\appcontentsline{app:sec4_proofs}{Appendix D. Proofs for Section 4}
\appcontentsline{app:sec5_proofs}{Appendix E. Proofs for Section 5}
\appcontentsline{app:sec6_proofs}{Appendix F. Details and proofs for Section 6}
\appcontentsline{app:secA_proofs}{Appendix G. Proofs for Appendix A}
%\clearpage

\section{Extension to joint risk aversion}
\label{app:joint_concavity}

Risk aversion at both stages can be handled by applying
the main analysis to transformed outcomes.
The function $\phi$ evaluates outcomes across population
members, while $\psi$ evaluates how the resulting policy
value varies across experimental samples.
The criterion retains the paper's two-agent interpretation,
or can represent a single social evaluation that
distinguishes the two sources of uncertainty.
No Taylor expansion of $\phi$ is needed.
The mean-variance approximation, when used, applies
only to the outer utility $\psi$.

The transformed treatment contrast becomes the relevant
target for policy learning.
Retain the sampling and randomized-treatment design
of the main text.
Fix a known increasing and concave function $\phi$,
finite on an interval containing the potential outcomes,
such that
$
\E|\phi\{Y(0)\}|+\E|\phi\{Y(1)\}|<\infty.
$
Recall
$
V_\phi(G)=\E\{\phi(Y^G)\}$ and $
\tau_\phi(x)
=
\E\left[
\phi\{Y(1)\}-\phi\{Y(0)\}\mid X=x
\right].
$
Stratified randomization identifies $\tau_\phi$ by
comparing conditional means of the transformed observed
outcome $\phi(Y)$ across treatment arms.
The corresponding oracle is
$G_\phi^\star(x)=\1\{\tau_\phi(x)\geq0\}$.

Risk-aware regret now measures the expected loss and
sampling variance of transformed welfare.
For $\rho\geq0$, define
\[
R_{\phi,\RA}(A)
=
V_\phi(G_\phi^\star)
-
\E[V_\phi\{A(S)\}]
+
\rho\V[V_\phi\{A(S)\}].
\]
For computational bagging, the expectation and variance
include both the training sample and the bag draws.
Given an estimator $\widehat\tau_{\phi,I}$, define
the bag-specific policy by
$\widehat G_I(x)=\1\{\widehat\tau_{\phi,I}(x)\geq0\}$,
and let
$
Q_{\phi,m,p}
=
\E_{S,I,X}\left\{
|\widehat\tau_{\phi,I}(X)-\tau_\phi(X)|^p
\right\}$ for $ p>1.
$
Complete and computational policy-vote bagging average
these policies as in the main text.

\begin{theorem}[Joint risk aversion]
\label[theorem]{thm:joint_concavity}
Let $\psi$ be increasing and concave on an open interval
containing the feasible transformed welfare values.
Suppose that, for some $\kappa>0$ and $C_\phi<\infty$, the generalized margin condition
$
\Pp_X\{0<|\tau_\phi(X)|\leq t\}
\leq C_\phi t^\kappa,
$ holds for $t>0
$
The main-text results transfer under their corresponding
assumptions, with
$(W,\tau,\widehat\tau_I,Q_{m,p},G^\star)$ replaced by
$(V_\phi,\tau_\phi,\widehat\tau_{\phi,I},
Q_{\phi,m,p},G_\phi^\star)$.
All welfare-based quantities, including regret,
certainty equivalents, and stability parameters,
are computed using $V_\phi$.

The averaging comparison in \Cref{prop:bag_concave}
requires only the stated concavity of $\psi$.
The certainty-equivalent expansion and its consequences
require the additional conditions of \Cref{lem:taylor},
imposed on the transformed welfare interval.
The exact CARA results apply to CARA outer utility,
with welfare range
$\lVert\tau_\phi\rVert_{L_1(P_X)}$.

In particular, suppose $Q_{\phi,m,2}<\infty$ and use
deterministic symmetric base learners.
For complete bagging and for computational bagging
with independent uniform bag draws,
\[
\begin{aligned}
R_{\phi,\RA}(A_m)
&\lesssim
Q_{\phi,m,2}^{\frac{\kappa+1}{\kappa+2}}
+
\rho\frac mn Q_{\phi,m,2},\\
R_{\phi,\RA}(A_{m,B})
&\lesssim
Q_{\phi,m,2}^{\frac{\kappa+1}{\kappa+2}}
+
\rho\left\{
\frac mn+\frac{1-m/n}{B}
\right\}Q_{\phi,m,2}.
\end{aligned}
\]
If $Q_{\phi,m,2}\lesssim m^{-\alpha}$ uniformly in $m$,
then
$
R_{\phi,\RA}(A_m)
\lesssim
m^{-\frac{\alpha(\kappa+1)}{\kappa+2}}
+
\frac{\rho}{n}m^{1-\alpha}.
$
This rate is sharp over the high-level
class whenever $\phi$ is strictly increasing on a
nondegenerate compact interval, and the fixed class
bounds admit a rescaled lower-bound construction.
\end{theorem}

In summary, joint risk aversion changes the outcome scale and
the treatment contrast that the policy-learning algorithm estimates.
The connection between stability and sampling risk
remains the same: it is applied to the population
value $V_\phi$ of the learned policy.

\section{Proofs for Section 2}\label{app:sec2_proofs}

\begin{proof}[Proof of \Cref{lem:taylor}]
We proceed in steps.
\begin{enumerate}
\item We begin with bounds that are uniform over algorithms.
Assumption~\ref{ass:main}(i) and $0\leq G\leq1$ place all feasible
welfare values in a common compact subinterval of $\mathcal J$.
Monotonicity of $\psi$ places their certainty equivalents in the
same interval. On this interval, $\psi'$ is bounded away from
zero, and $\psi''$ and $\psi'''$ are bounded. All
suprema and infima are over this interval.
We use $C_\psi$ for a sufficiently large finite constant
depending only on these bounds and the interval. It may change from line to line.

\item We expand expected utility around mean welfare.
For each realization of $Z_A$, Taylor's formula gives
\[
\begin{aligned}
\psi(Z_A)
={}&\psi\{\E(Z_A)\}
+\psi'\{\E(Z_A)\}\{Z_A-\E(Z_A)\}\\
&+\frac12\psi''\{\E(Z_A)\}\{Z_A-\E(Z_A)\}^2
+\text{remainder},
\end{aligned}
\]
where the absolute remainder is bounded by
$C_\psi|Z_A-\E(Z_A)|^3$.
Taking expectations eliminates the linear term because its
coefficient is nonrandom and $\E\{Z_A-\E(Z_A)\}=0$.
The squared deviation in the quadratic term becomes the variance.
Thus,
\[
\E\{\psi(Z_A)\}-\psi\{\E(Z_A)\}
=
\frac12\psi''\{\E(Z_A)\}\V(Z_A)
+\text{error}_1,
\]
where $\text{error}_1$ is the expected remainder.
The triangle inequality gives
$|\text{error}_1|\leq C_\psi\E|Z_A-\E(Z_A)|^3$.

\item Next, we bound the distance between the certainty
equivalent and mean welfare.
Since $Z_A$ and $\E(Z_A)$ lie in the common compact interval,
$|Z_A-\E(Z_A)|$ is uniformly bounded. Consequently,
\[
\E|Z_A-\E(Z_A)|^3
\leq C_\psi\V(Z_A).
\]
Together with the bound on $\psi''$, Step~2 therefore implies
\[
\left|\E\{\psi(Z_A)\}-\psi\{\E(Z_A)\}\right|
\leq C_\psi\V(Z_A).
\]
The positive lower bound on $\psi'$ converts this utility gap
into a welfare gap. By the mean value theorem and the definition
$\psi\{\CE_\psi(A)\}=\E\{\psi(Z_A)\}$,
\[
\left|\CE_\psi(A)-\E(Z_A)\right|
\leq
\frac{
\left|\E\{\psi(Z_A)\}-\psi\{\E(Z_A)\}\right|
}{\inf\psi'}
\leq C_\psi\V(Z_A).
\]
For concave $\psi$, this distance is the amount of mean welfare
the researcher would give up to eliminate sampling risk.

\item Finally, we obtain the certainty-equivalent expansion.
A first-order Taylor expansion of $\psi\{\CE_\psi(A)\}$
around $\E(Z_A)$ gives
\[
\psi\{\CE_\psi(A)\}-\psi\{\E(Z_A)\}
=
\psi'\{\E(Z_A)\}\{\CE_\psi(A)-\E(Z_A)\}
+\text{error}_2.
\]
The bound on $\psi''$ and Step~3 imply
\[
|\text{error}_2|
\leq C_\psi\{\CE_\psi(A)-\E(Z_A)\}^2
\leq C_\psi\V(Z_A)^2.
\]

The definition of the certainty equivalent gives
$\psi\{\CE_\psi(A)\}=\E\{\psi(Z_A)\}$.
Thus this expansion and the expansion in Step~2 describe
the same utility gap, and their right-hand sides must agree:
\[
\psi'\{\E(Z_A)\}\{\CE_\psi(A)-\E(Z_A)\}
+\text{error}_2
=
\frac12\psi''\{\E(Z_A)\}\V(Z_A)
+\text{error}_1.
\]
We now solve for $\CE_\psi(A)-\E(Z_A)$.
Subtracting $\text{error}_2$ from both sides gives
\[
\psi'\{\E(Z_A)\}\{\CE_\psi(A)-\E(Z_A)\}
=
\frac12\psi''\{\E(Z_A)\}\V(Z_A)
+\text{error}_1-\text{error}_2.
\]
Dividing by the positive coefficient $\psi'\{\E(Z_A)\}$
therefore yields
\[
\begin{aligned}
\CE_\psi(A)-\E(Z_A)
&=
\frac{\psi''\{\E(Z_A)\}}{2\psi'\{\E(Z_A)\}}\V(Z_A)
+
\frac{\text{error}_1-\text{error}_2}
     {\psi'\{\E(Z_A)\}}\\
&=
-\frac{\gamma}{2}\V(Z_A)+r_\psi(A),
\end{aligned}
\]
where we use $-\psi''/\psi'=\gamma$ and identify $r_\psi(A)$
with the fraction containing the two errors.
The positive lower bound on $\psi'$ and the triangle inequality
therefore give
\[
|r_\psi(A)|
\leq C_\psi
\left\{
\E|Z_A-\E(Z_A)|^3+\V(Z_A)^2
\right\}.
\]

Both terms vanish uniformly when the centered third moments
vanish uniformly. Indeed,
$\V(Z_A)^2
\leq\{\E|Z_A-\E(Z_A)|^3\}^{4/3}$.
Since all constants are independent of $A$,
$\sup_{A\in\mathcal A_n}\E|Z_A-\E(Z_A)|^3\to0$
implies
$\sup_{A\in\mathcal A_n}|r_\psi(A)|\to0$, as claimed.
\end{enumerate}
\end{proof}

\begin{proof}[Proof of \Cref{cor:optimizer}]
We proceed in steps. Throughout, we use
$\CE_\psi(A)=M_\rho(A)+r_\psi(A)$ and
$|r_\psi(A)|\leq\epsilon_n$ for every $A\in\mathcal A_n$.
\begin{enumerate}
\item We bound the loss in certainty equivalent from choosing
the mean-variance maximizer.
Substituting the expansion for each algorithm gives
\[
\begin{aligned}
\CE_\psi(A_n^{\CE})-\CE_\psi(A_n^{\mathrm{MV}})
={}&M_\rho(A_n^{\CE})-M_\rho(A_n^{\mathrm{MV}})
+r_\psi(A_n^{\CE})-r_\psi(A_n^{\mathrm{MV}}).
\end{aligned}
\]
The first difference on the right is nonpositive because
$A_n^{\mathrm{MV}}$ maximizes $M_\rho$.
The difference on the left is nonnegative because
$A_n^{\CE}$ maximizes $\CE_\psi$.
Consequently,
\[
0
\leq
\CE_\psi(A_n^{\CE})-\CE_\psi(A_n^{\mathrm{MV}})
\leq
r_\psi(A_n^{\CE})-r_\psi(A_n^{\mathrm{MV}})
\leq 2\epsilon_n,
\]
where the last inequality follows because each remainder
has absolute value at most $\epsilon_n$.

\item We bound the loss in mean-variance value from choosing
the certainty-equivalent maximizer.
Using $M_\rho(A)=\CE_\psi(A)-r_\psi(A)$ gives
\[
\begin{aligned}
M_\rho(A_n^{\mathrm{MV}})-M_\rho(A_n^{\CE})
={}&\CE_\psi(A_n^{\mathrm{MV}})-\CE_\psi(A_n^{\CE})
+r_\psi(A_n^{\CE})-r_\psi(A_n^{\mathrm{MV}}).
\end{aligned}
\]
The first difference on the right is nonpositive by
certainty-equivalent optimality.
The difference on the left is nonnegative by mean-variance
optimality. Applying the same remainder bound therefore gives
\[
0
\leq
M_\rho(A_n^{\mathrm{MV}})-M_\rho(A_n^{\CE})
\leq
r_\psi(A_n^{\CE})-r_\psi(A_n^{\mathrm{MV}})
\leq 2\epsilon_n.
\]

\item We show that a sufficiently large objective gap forces
the maximizers to coincide.
Suppose $A_n^{\mathrm{MV}}$ is the unique mean-variance maximizer
and every other candidate $A\in\mathcal A_n$ satisfies
\[
M_\rho(A_n^{\mathrm{MV}})-M_\rho(A)>2\epsilon_n.
\]
If $A_n^{\CE}$ were a different candidate, this inequality
would contradict the bound in Step~2.
Thus $A_n^{\CE}=A_n^{\mathrm{MV}}$.
Since the argument applies to every certainty-equivalent
maximizer, that maximizer is also unique.
\end{enumerate}
\end{proof}

\section{Proofs for Section 3}\label{app:sec3_proofs}

For $\rho>0$, holding risk-aware regret
$R(A)+\rho\V(Z_A)$ constant gives
$\dd R(A)+\rho\,\dd\V(Z_A)=0$.
Thus lines of constant risk-aware regret have slope
\[
\frac{\dd\V(Z_A)}{\dd R(A)}=-\frac1\rho.
\]
Where the efficient boundary expresses variance as a
differentiable function of expected regret, an interior
optimum satisfies
\[
\frac{\dd}{\dd R(A)}
\left\{R(A)+\rho\V(Z_A)\right\}
=
1+\rho\frac{\dd\V(Z_A)}{\dd R(A)}
=0.
\]
Therefore, at a smooth interior optimum, the efficient boundary
has the same slope as the line of constant risk-aware regret.
\section{Proofs for Section 4}\label{app:sec4_proofs}

\subsection{Welfare stability}

\begin{proof}[Proof of \Cref{lem:policy_to_welfare}]
We proceed in steps.
\begin{enumerate}
\item We express the welfare difference in terms of treatment
probabilities.
Fix $S,i,s'$. By $W(G)=
\E\{Y(0)\}+\E\{\tau(X)G(X)\}$, the common baseline
$\mathbb E\{Y(0)\}$ cancels, leaving
\[
W\{A(S)\}-W\{A(S^{i\leftarrow s'})\}
=
\mathbb E_X
\left(
\tau(X)
\left[\{A(S)\}(X)-\{A(S^{i\leftarrow s'})\}(X)\right]
\right).
\]

\item We bound this difference using the algorithmic
stability parameter $\Gamma_n(A)$.
By definition,
$|\{A(S)\}(x)-\{A(S^{i\leftarrow s'})\}(x)|\leq\Gamma_n(A)$
for every $x\in\mathcal X$.
Hence
\[
\begin{aligned}
\left|W\{A(S)\}-W\{A(S^{i\leftarrow s'})\}\right|
&\leq
\mathbb E_X
\left[
|\tau(X)|
\left|\{A(S)\}(X)-\{A(S^{i\leftarrow s'})\}(X)\right|
\right]\\
&\leq
\Gamma_n(A)\mathbb E_X|\tau(X)|\\
&=
\Gamma_n(A)\|\tau\|_{L_1(P_X)}.
\end{aligned}
\]
We can take $\Gamma_n(A)$ outside of the expectation because it
does not depend on the population covariate $X$.
The resulting bound holds for every $S,i,s'$, so taking their
supremum gives
$\beta_n(A)\leq\Gamma_n(A)\|\tau\|_{L_1(P_X)}$.

\item For a score-based rule, Lipschitz continuity converts
score stability into algorithmic stability.
Suppose $\{A(S)\}(x)=\delta\{\zeta(S,x)\}$, where $\delta$ is
$L_\delta$-Lipschitz and replacing one observation changes
the score by at most $s_n$, uniformly over $S,i,s',x$.
Then
\[
\begin{aligned}
\left|\{A(S)\}(x)-\{A(S^{i\leftarrow s'})\}(x)\right|
&=
\left|
\delta\{\zeta(S,x)\}
-\delta\{\zeta(S^{i\leftarrow s'},x)\}
\right|\\
&\leq
L_\delta
\left|\zeta(S,x)-\zeta(S^{i\leftarrow s'},x)\right|\\
&\leq L_\delta s_n.
\end{aligned}
\]
Taking the supremum over $S,i,s',x$ gives
$\Gamma_n(A)\leq L_\delta s_n$.
Substituting into the bound from Step~2 yields
$
\beta_n(A)
\leq L_\delta s_n\|\tau\|_{L_1(P_X)}.
$
\end{enumerate}
\end{proof}

\subsection{Main result}

\begin{lemma}[McDiarmid bounds for welfare]
\label{lem:mcdiarmid_welfare}
Let $Z_A=W\{A(S)\}$, where the observations in $S$ are
independent and replacing any one observation changes
$Z_A$ by at most $\beta_n(A)<\infty$.
Then, for every $u\in\mathbb R$,
\[
\E\left(
\exp\left[u\{Z_A-\E(Z_A)\}\right]
\right)
\leq
\exp\left\{
\frac{u^2n}{8}\beta_n(A)^2
\right\}.
\]
For $\beta_n(A)>0$ and every $t>0$,
\[
\Pr\{\E(Z_A)-Z_A\geq t\}
\leq
\exp\left\{
-\frac{2t^2}{n\beta_n(A)^2}
\right\},\quad 
\Pr\{|Z_A-\E(Z_A)|\geq t\}
\leq
2\exp\left\{
-\frac{2t^2}{n\beta_n(A)^2}
\right\}.
\]
Moreover, for a universal finite constant $C$,
$
\E|Z_A-\E(Z_A)|^3
\leq C\{n\beta_n(A)^2\}^{3/2}.
$
If $\beta_n(A)=0$, then $Z_A$ is constant.
We interpret both tail bounds as zero in this case.
\end{lemma}

\begin{proof}
If $\beta_n(A)=0$, replacing observations one at a time
never changes welfare, so $Z_A$ is constant.
The exponential moment is then one, while the centered
third moment and both tail probabilities are zero.
For $\beta_n(A)>0$, we proceed in steps.
\begin{enumerate}
\item We apply the exponential bounded-differences inequality.
The observations are independent, and replacing each
observation changes welfare by at most $\beta_n(A)$.
The exponential form of McDiarmid's inequality
\citep{mcdiarmid1989} therefore gives, for every
$u\in\mathbb R$,
\[
\begin{aligned}
\E\left(
\exp\left[u\{Z_A-\E(Z_A)\}\right]
\right)
&\leq
\exp\left\{
\frac{u^2}{8}\sum_{i=1}^n\beta_n(A)^2
\right\}
=
\exp\left\{
\frac{u^2n}{8}\beta_n(A)^2
\right\}.
\end{aligned}
\]
The equality uses the same replacement bound for all
$n$ observations.

\item We turn the exponential-moment bound into tail bounds.
Fix $t>0$ and $u>0$.
Because the exponential is increasing,
$\E(Z_A)-Z_A\geq t$ exactly when
$\exp[u\{\E(Z_A)-Z_A\}]\geq e^{ut}$.
Markov's inequality bounds this probability by the
expectation of the nonnegative exponential divided
by its threshold:
\[
\begin{aligned}
\Pr\{\E(Z_A)-Z_A\geq t\}
&\leq
e^{-ut}\,
\E\left(
\exp\left[u\{\E(Z_A)-Z_A\}\right]
\right)\\
&=
e^{-ut}\,
\E\left(
\exp\left[-u\{Z_A-\E(Z_A)\}\right]
\right).
\end{aligned}
\]

Substituting Step~1 with parameter $-u$ gives
\[
\Pr\{\E(Z_A)-Z_A\geq t\}
\leq
\exp\left\{
-ut+\frac{u^2n}{8}\beta_n(A)^2
\right\}.
\]
The exponent is a quadratic function of $u$ with positive
quadratic coefficient. Its derivative is
$-t+un\beta_n(A)^2/4$, so its minimum occurs at
$u=4t/\{n\beta_n(A)^2\}$.
Substituting,
\[
\Pr\{\E(Z_A)-Z_A\geq t\}
\leq
\exp\left\{
-\frac{2t^2}{n\beta_n(A)^2}
\right\}.
\]

For the upper tail, apply Markov's inequality to
$\exp[u\{Z_A-\E(Z_A)\}]$ and use Step~1 with parameter $u$.
The resulting quadratic exponent is the same, so the
upper-tail probability has the same bound.
Since $|Z_A-\E(Z_A)|\geq t$ occurs when either tail event
occurs, adding their bounds gives
\[
\Pr\{|Z_A-\E(Z_A)|\geq t\}
\leq
2\exp\left\{
-\frac{2t^2}{n\beta_n(A)^2}
\right\}.
\]

\item We express the third moment as an integral
of tail probabilities.
For each realization, $|Z_A-\E(Z_A)|$ is a nonnegative
number. Integrating $3t^2$ up to this number gives
\[
\int_0^{|Z_A-\E(Z_A)|}3t^2\,\dd t
=
\left.t^3\right|_0^{|Z_A-\E(Z_A)|}
=
|Z_A-\E(Z_A)|^3.
\]
We can instead integrate over all $t\geq0$ by including
an indicator that restricts the integral to values below
the realized absolute deviation:
\[
|Z_A-\E(Z_A)|^3
=
\int_0^\infty
3t^2\mathbf 1\{t<|Z_A-\E(Z_A)|\}\,\dd t.
\]
The integrand is nonnegative, so Tonelli's theorem allows
us to interchange expectation and integration:
\[
\begin{aligned}
\E|Z_A-\E(Z_A)|^3
&=
\E\left[
\int_0^\infty
3t^2\mathbf 1\{t<|Z_A-\E(Z_A)|\}\,\dd t
\right]\\
&=
3\int_0^\infty
t^2\E\left[
\mathbf 1\{t<|Z_A-\E(Z_A)|\}
\right]\,\dd t\\
&=
3\int_0^\infty
t^2\Pr\{t<|Z_A-\E(Z_A)|\}\,\dd t.
\end{aligned}
\]

\item We integrate the tail bound to obtain the
third-moment bound.
Substituting the two-sided bound from Step~2 into
the identity from Step~3 gives
\[
\E|Z_A-\E(Z_A)|^3
\leq
6\int_0^\infty
t^2\exp\left\{
-\frac{2t^2}{n\beta_n(A)^2}
\right\}\,\dd t.
\]
Set $t=\sqrt{n}\beta_n(A)v$, so that
$\dd t=\sqrt{n}\beta_n(A)\,\dd v$.
The squared term contributes a factor $n\beta_n(A)^2$,
and the change in the integration variable contributes
another factor $\sqrt{n}\beta_n(A)$:
\[
\E|Z_A-\E(Z_A)|^3
\leq
6\{n\beta_n(A)^2\}^{3/2}
\int_0^\infty v^2e^{-2v^2}\,\dd v.
\]
The remaining integral is finite and independent of
$A$ and $n$. Absorbing it and the factor of six into
a universal constant $C$ proves
$\E|Z_A-\E(Z_A)|^3
\leq C\{n\beta_n(A)^2\}^{3/2}$.
\end{enumerate}
\end{proof}

\begin{lemma}[Variance bound on an interval \citep{popoviciu1935}]
\label{lem:interval_variance}
Let $Z$ be a real-valued random variable satisfying
$a\leq Z\leq b$ almost surely, where
$-\infty<a\leq b<\infty$. Then
$
\V(Z)\leq \frac{(b-a)^2}{4}.
$
\end{lemma}

\begin{proof}[Proof of \Cref{thm:stability}]
We proceed in steps, writing $Z_A=W\{A(S)\}$.
The integrability assumption bounds population welfare uniformly
over $S$, so $Z_A$ has finite variance.
\begin{enumerate}
\item We relate welfare variance to the effect of replacing
one observation.
Let $S_{-i}$ denote the sample with observation $i$ omitted,
and let $S_i'$ be an independent copy of $S_i$, independent
of the full sample.
Conditional on $S_{-i}$, the welfare values $Z_A$ and
$W\{A(S^{i\leftarrow S_i'})\}$ are independent and have the
same distribution. Expanding their squared difference gives
\[
\begin{aligned}
\mathbb E\left(
\left[Z_A-W\{A(S^{i\leftarrow S_i'})\}\right]^2
\mid S_{-i}
\right)
&=
2\mathbb E(Z_A^2\mid S_{-i})
-2\{\mathbb E(Z_A\mid S_{-i})\}^2\\
&=2\V(Z_A\mid S_{-i}).
\end{aligned}
\]
Here conditional independence makes the expected cross-product
equal to the squared conditional mean.
The Efron-Stein inequality gives
\[
\begin{aligned}
\V(Z_A)
&\leq
\frac12\sum_{i=1}^n
\mathbb E\left(
\left[Z_A-W\{A(S^{i\leftarrow S_i'})\}\right]^2
\right)\\
&=
\frac12\sum_{i=1}^n
\mathbb E\left\{
\mathbb E\left(
\left[Z_A-W\{A(S^{i\leftarrow S_i'})\}\right]^2
\mid S_{-i}
\right)
\right\}\\
&=
\frac12\sum_{i=1}^n
\mathbb E\{2\V(Z_A\mid S_{-i})\}\\
&=
\sum_{i=1}^n
\mathbb E\{\V(Z_A\mid S_{-i})\}.
\end{aligned}
\]
The first equality uses the law of iterated expectations. 
The second equality substitutes the identity established above.
The last equality cancels factors.

\item We apply the interval variance bound to conditional welfare.
Fix $i$ and hold $S_{-i}$ constant.
Any two welfare values obtained by varying $S_i$ differ by
at most $\beta_n(A)$, by the definition of welfare stability.
Therefore, the conditional distribution of $Z_A$ lies in an
interval of length at most $\beta_n(A)$.
The interval may depend on $S_{-i}$, but the bound on its
length is uniform.

Applying \Cref{lem:interval_variance} to the conditional
distribution gives
$
\V(Z_A\mid S_{-i})
\leq \frac{\beta_n(A)^2}{4}
$ almost surely.
Taking expectations and substituting into Step~1 yields
\[
\begin{aligned}
\V(Z_A)
&\leq
\sum_{i=1}^n
\mathbb E\{\V(Z_A\mid S_{-i})\}
\leq
\sum_{i=1}^n\frac{\beta_n(A)^2}{4}
=
\frac{n}{4}\beta_n(A)^2.
\end{aligned}
\]
Finally, the definition of risk-aware regret implies
\[
R_{\mathrm{RA}}(A)
=
R(A)+\rho\V(Z_A)
\leq
R(A)+\rho\,\frac{n}{4}\beta_n(A)^2.
\]

\item We control the probability that welfare falls below
its expectation.
The lower-tail bound in \Cref{lem:mcdiarmid_welfare} gives, for $t>0$,
\[
\Pr\{\E(Z_A)-Z_A\geq t\}
\leq
\exp\left\{
-\frac{2t^2}{n\beta_n(A)^2}
\right\}.
\]
Here $\E(Z_A)$ averages over possible experimental samples,
while $Z_A$ is the welfare produced by the observed sample.

\item We translate the welfare shortfall into excess regret.
Realized regret is $W^\star-Z_A$, while expected regret is
$R(A)=W^\star-\mathbb E(Z_A)$.
Subtracting gives
\[
\begin{aligned}
W^\star-Z_A-R(A)
&=
W^\star-Z_A-\{W^\star-\mathbb E(Z_A)\}
=
\mathbb E(Z_A)-Z_A.
\end{aligned}
\]
Using $Z_A=W\{A(S)\}$ and applying Step~3,
\[
\begin{aligned}
\Pr\left[
W^\star-W\{A(S)\}\geq R(A)+t
\right]
&\leq
\exp\left\{
-\frac{2t^2}{n\beta_n(A)^2}
\right\}.
\end{aligned}
\]
\end{enumerate}
\end{proof}

\subsection{Utility interpretation}

\begin{lemma}[CARA decomposition]
\label{lem:cara_decomposition}
Let $\gamma>0$, and suppose
$\E|Z_A|<\infty$ and $\E\{\exp(-\gamma Z_A)\}<\infty$.
Then
\[
\CE_\gamma(A)
=
\E(Z_A)
-
\frac1\gamma
\log\E\left(
\exp\left[-\gamma\{Z_A-\E(Z_A)\}\right]
\right).
\]
Moreover, $\CE_\gamma(A)\leq\E(Z_A)$.
Consequently,
$W^\star-\CE_\gamma(A)\geq R(A)$.
\end{lemma}

\begin{proof}
We proceed in steps.
\begin{enumerate}
\item We separate mean welfare from sampling variation.
Writing $Z_A=\E(Z_A)+\{Z_A-\E(Z_A)\}$,
\[
\E\{\exp(-\gamma Z_A)\}
=
\exp\{-\gamma\E(Z_A)\}
\E\left(
\exp\left[-\gamma\{Z_A-\E(Z_A)\}\right]
\right).
\]
The first factor can be taken outside the expectation
because $\E(Z_A)$ is nonrandom.
Taking the logarithm turns the product into a sum.
Substituting into the definition of the certainty
equivalent therefore gives
\[
\begin{aligned}
\CE_\gamma(A)
&=
-\frac1\gamma
\left\{
-\gamma\E(Z_A)
+
\log\E\left(
\exp\left[-\gamma\{Z_A-\E(Z_A)\}\right]
\right)
\right\}\\
&=
\E(Z_A)
-
\frac1\gamma
\log\E\left(
\exp\left[-\gamma\{Z_A-\E(Z_A)\}\right]
\right).
\end{aligned}
\]

\item We show that the certainty equivalent cannot
exceed mean welfare.
The exponential function is convex, so Jensen's
inequality gives
\[
\E\left(
\exp\left[-\gamma\{Z_A-\E(Z_A)\}\right]
\right)
\geq
\exp\left[-\gamma\E\{Z_A-\E(Z_A)\}\right]
=1.
\]
The equality uses $\E\{Z_A-\E(Z_A)\}=0$.
The logarithm in Step~1 is therefore nonnegative.
Since $\gamma>0$, subtracting this term gives
$\CE_\gamma(A)\leq\E(Z_A)$.
Subtracting both sides from $W^\star$ then yields
$W^\star-\CE_\gamma(A)
\geq W^\star-\E(Z_A)=R(A)$.
\end{enumerate}
\end{proof}

\begin{proof}[Proof of \Cref{prop:cara_stability}]
We proceed in steps, with $\gamma>0$.
\begin{enumerate}
\item We use welfare stability to bound the exponential
moment and its logarithm.
The exponential-moment bound in
\Cref{lem:mcdiarmid_welfare} holds for every real $u$.
Taking $u=-\gamma$ gives
\[
\E\left(
\exp\left[-\gamma\{Z_A-\E(Z_A)\}\right]
\right)
\leq
\exp\left\{
\frac{\gamma^2n}{8}\beta_n(A)^2
\right\}.
\]
Taking logarithms preserves the inequality.
Dividing by $\gamma>0$ then gives
\[
\frac1\gamma
\log\E\left(
\exp\left[-\gamma\{Z_A-\E(Z_A)\}\right]
\right)
\leq
\frac{\gamma n}{8}\beta_n(A)^2.
\]

\item We translate this bound into guarantees for the
certainty equivalent and regret.
The decomposition in \Cref{lem:cara_decomposition}
subtracts the logarithmic term from mean welfare.
Its upper bound from Step~1 therefore gives a lower
bound on the certainty equivalent:
\[
\begin{aligned}
\CE_\gamma(A)
&=
\E(Z_A)
-
\frac1\gamma
\log\E\left(
\exp\left[-\gamma\{Z_A-\E(Z_A)\}\right]
\right)
\geq
\E(Z_A)-\frac{\gamma n}{8}\beta_n(A)^2.
\end{aligned}
\]
Subtracting this inequality from $W^\star$ reverses
its direction:
\[
\begin{aligned}
W^\star-\CE_\gamma(A)
&\leq
W^\star-\E(Z_A)+\frac{\gamma n}{8}\beta_n(A)^2
=
R(A)+\frac{\gamma n}{8}\beta_n(A)^2.
\end{aligned}
\]
\end{enumerate}
\end{proof}

\begin{proof}[Proof of \Cref{prop:cara_variance}]
We proceed in steps, with $\gamma>0$.
\begin{enumerate}
\item We bound how far welfare can fall below its expectation.
Fix the experimental sample $S$.
At covariate value $x$, the oracle delivers a gain of
$\max\{\tau(x),0\}$ relative to no treatment.
The learned policy delivers the conditional expected gain
$\tau(x)\{A(S)\}(x)$.
Since $\{A(S)\}(x)\in[0,1]$, their difference satisfies
\[
0\leq
\max\{\tau(x),0\}-\tau(x)\{A(S)\}(x)
\leq |\tau(x)|.
\]
Taking expectations over the population covariate $X$,
while holding $S$ fixed, gives
\[
\begin{aligned}
0\leq W^\star-Z_A
&=
\E_X\left[
\max\{\tau(X),0\}-\tau(X)\{A(S)\}(X)
\right]\\
&\leq
\E_X|\tau(X)|=T.
\end{aligned}
\]
This bound holds for every experimental sample.

If $T=0$, then $Z_A=W^\star$ almost surely, so expected
regret, welfare variance, and the certainty-equivalent
shortfall all vanish. Henceforth, suppose $T>0$.
Since $\E(Z_A)\leq W^\star$, we also have
$\E(Z_A)-Z_A\leq W^\star-Z_A\leq T$.

\item We bound the exponential by a quadratic expression.
For every real $x\leq T$, a first-order Taylor expansion
of $e^{\gamma x}$ around zero, with its exact integral
remainder, gives
\[
e^{\gamma x}
=
1+\gamma x+
\gamma^2x^2\int_0^1(1-u)e^{\gamma ux}\,\dd u.
\]
Because $\gamma>0$ and $u\in[0,1]$, the inequality $x\leq T$
implies $e^{\gamma ux}\leq e^{\gamma uT}$.
Thus
\[
\begin{aligned}
e^{\gamma x}
&\leq
1+\gamma x+
\gamma^2x^2\int_0^1(1-u)e^{\gamma uT}\,\dd u\\
&=
1+\gamma x+
\frac{e^{\gamma T}-1-\gamma T}{T^2}x^2.
\end{aligned}
\]
The last equality uses integration by parts.
The inequality also holds for negative $x$, because
$\gamma^2x^2\geq0$.

\item We use welfare variance to bound the exponential
moment and its logarithm.
Step~1 allows us to substitute $x=\E(Z_A)-Z_A$ into
the bound from Step~2.
Taking expectations eliminates the linear term because
$\E\{\E(Z_A)-Z_A\}=0$.
The expected squared term is $\V(Z_A)$.
Consequently,
\[
\E\left(
\exp\left[\gamma\{\E(Z_A)-Z_A\}\right]
\right)
\leq
1+
\frac{e^{\gamma T}-1-\gamma T}{T^2}\V(Z_A).
\]

Taking logarithms preserves this inequality, as does
division by $\gamma>0$.
We then apply $\log(1+y)\leq y$ to obtain
\[
\begin{aligned}
\frac1\gamma
\log\E\left(
\exp\left[\gamma\{\E(Z_A)-Z_A\}\right]
\right)
&\leq
\frac1\gamma
\log\left\{
1+\frac{e^{\gamma T}-1-\gamma T}{T^2}\V(Z_A)
\right\}\\
&\leq
\frac{e^{\gamma T}-1-\gamma T}{\gamma T^2}\V(Z_A).
\end{aligned}
\]
Here the quantity added to one is nonnegative because
$e^{\gamma T}\geq1+\gamma T$ and $\V(Z_A)\geq0$.

\item We translate this bound into a
certainty-equivalent shortfall.
By \Cref{lem:cara_decomposition}, using
$-\gamma\{Z_A-\E(Z_A)\}
=\gamma\{\E(Z_A)-Z_A\}$ and
$R(A)=W^\star-\E(Z_A)$,
\[
\begin{aligned}
W^\star-\CE_\gamma(A)
&=
R(A)+\frac1\gamma
\log\E\left(
\exp\left[\gamma\{\E(Z_A)-Z_A\}\right]
\right)\\
&\leq
R(A)+
\frac{e^{\gamma T}-1-\gamma T}{\gamma T^2}\V(Z_A).
\end{aligned}
\]
The inequality substitutes the bound from Step~3.
The lower bound
$W^\star-\CE_\gamma(A)\geq R(A)$
is also established by the lemma.
\end{enumerate}
\end{proof}

\begin{proof}[Proof of \Cref{cor:stable_taylor}]
We proceed in steps.
\begin{enumerate}
\item We bound the centered third moments uniformly
over algorithms.
Since $\beta_n(A)\leq\bar\beta_n$ for every $A\in\cA_n$,
the third-moment bound in \Cref{lem:mcdiarmid_welfare} gives
\[
\sup_{A\in\cA_n}
\E|Z_A-\E(Z_A)|^3
\leq C(n\bar\beta_n^2)^{3/2}.
\]
The constant $C$ is independent of $A$ and $n$.

\item We substitute the moment bounds into the
certainty-equivalent remainder.
By \Cref{thm:stability},
$\V(Z_A)\leq n\bar\beta_n^2/4$ uniformly over $A\in\cA_n$.
Squaring gives
$\V(Z_A)^2\leq(n\bar\beta_n^2)^2/16$.
Together with Step~1, the remainder bound in
\Cref{lem:taylor} therefore implies
\[
\sup_{A\in\cA_n}|r_\psi(A)|
\leq
C_\psi
\left\{
(n\bar\beta_n^2)^{3/2}
+
(n\bar\beta_n^2)^2
\right\},
\]
after absorbing the numerical constants into $C_\psi$.
Both terms vanish when $n\bar\beta_n^2\to0$,
which gives the claimed uniform approximation.
\end{enumerate}
\end{proof}
\section{Proofs for Section 5}\label{app:sec5_proofs}

\subsection{Policy-vote bagging}

\begin{proof}[Proof of \Cref{prop:bag_concave}]
We proceed in steps.
\begin{enumerate}
\item Computational and complete bagging
have the same conditional mean welfare.
Fix the training sample $S$.
Welfare is affine in the policy, so averaging policies
also averages their welfare:
\[
\begin{aligned}
W\{A_{m,B}(S,\mathbf I)\}
&=
\frac1B\sum_{b=1}^B W(\widehat G_{I_b}),\quad 
W\{A_m(S)\}
=
\E_I\{W(\widehat G_I)\mid S\}.
\end{aligned}
\]
Conditional on $S$, each $I_b$ is uniform over the
size-$m$ bags. Its expected welfare therefore equals
that of a single uniform bag $I$.
Taking conditional expectations gives
\[
\begin{aligned}
\E_{\mathbf I}\left[
W\{A_{m,B}(S,\mathbf I)\}\mid S
\right]
&=
\frac1B\sum_{b=1}^B
\E_{I_b}\{W(\widehat G_{I_b})\mid S\}
&=
\E_I\{W(\widehat G_I)\mid S\}
=
W\{A_m(S)\}.
\end{aligned}
\]

\item Complete bagging yields at least
as much conditional expected utility as computational
bagging.
Step~1 expresses complete-bagging welfare as the
conditional mean of computational-bagging welfare.
Since $\psi$ is concave, conditional Jensen's
inequality gives
\[
\begin{aligned}
\psi[W\{A_m(S)\}]
&=
\psi\left(
\E_{\mathbf I}\left[
W\{A_{m,B}(S,\mathbf I)\}\mid S
\right]
\right)
\geq
\E_{\mathbf I}\left(
\psi\left[W\{A_{m,B}(S,\mathbf I)\}\right]
\mid S
\right).
\end{aligned}
\]
The averaging is over the random collection of
bags, with the training sample fixed.

\item Computational bagging yields at least
as much conditional expected utility as a single random bag.
Fix both $S$ and the realized collection $\mathbf I$.
Concavity applied to the average of the $B$ welfare
values gives
\[
\begin{aligned}
\psi\left[W\{A_{m,B}(S,\mathbf I)\}\right]
&=
\psi\left\{
\frac1B\sum_{b=1}^B W(\widehat G_{I_b})
\right\}
\geq
\frac1B\sum_{b=1}^B
\psi\{W(\widehat G_{I_b})\}.
\end{aligned}
\]
We now average over the random collection of bags.
Each $I_b$ has the same conditional distribution as
a single uniform bag $I$, so
\[
\begin{aligned}
\E_{\mathbf I}\left(
\psi\left[W\{A_{m,B}(S,\mathbf I)\}\right]
\mid S
\right)
&\geq
\frac1B\sum_{b=1}^B
\E_{I_b}\left[
\psi\{W(\widehat G_{I_b})\}\mid S
\right]
=
\E_I\left[
\psi\{W(\widehat G_I)\}\mid S
\right].
\end{aligned}
\]

\item We extend the comparisons to randomness in the
training sample.
Integrability of the potential outcomes places all
feasible welfare values in a fixed compact interval.
The assumptions on $\psi$ make its values bounded
on that interval, so all welfare and utility expectations
above are finite.
Taking expectations over $S$ and applying the law of
iterated expectations therefore preserves the mean
equality from Step~1 and both utility inequalities
from Steps~2 and~3.
\end{enumerate}
\end{proof}

\subsection{Policy-vote bagging is stable}

\begin{proof}[Proof of \Cref{lem:bag_stability}]
We proceed in steps.
\begin{enumerate}
\item We bound the change in treatment probabilities
by counting the affected bags.
Fix $S$, $i$, $s'$, and $x\in\cX$.
If a bag $I$ omits $i$, replacing $S_i$ leaves $S_I$
unchanged and therefore leaves its fitted policy unchanged.
If $I$ contains $i$, its original and replacement
binary policies can differ by at most one at $x$.

The complete policy averages over all $\binom{n}{m}$
size-$m$ bags.
Exactly $\binom{n-1}{m-1}$ of these contain $i$:
after including $i$, we choose the remaining $m-1$
indices from the other $n-1$ observations.
The triangle inequality therefore gives
\[
\begin{aligned}
\left|
\{A_m(S)\}(x)
-
\{A_m(S^{i\leftarrow s'})\}(x)
\right|
&\leq
\frac{1}{\binom{n}{m}}
\sum_{\substack{I\subseteq\{1,\ldots,n\}\\ |I|=m}}
\mathbf 1\{i\in I\}
=
\frac{\binom{n-1}{m-1}}{\binom{n}{m}}
=
\frac{m}{n}.
\end{aligned}
\]
Taking the supremum over $S,i,s',x$ gives
$\Gamma_n(A_m)\leq m/n$.

\item We convert the algorithmic stability bound into a welfare stability bound.
By \Cref{lem:policy_to_welfare},
\[
\beta_n(A_m)
\leq
\Gamma_n(A_m)\|\tau\|_{L_1(P_X)}
\leq
\frac{m}{n}\|\tau\|_{L_1(P_X)}.
\]

\item We construct a case that attains the bound.
Let $X=0$ almost surely, $Y(1)=1/2$, and $Y(0)=-1/2$,
with treatment randomized with probability one half.
Then $\tau(0)=1$ and $\|\tau\|_{L_1(P_X)}=1$.
Let $s^+$ denote the observed sample value for a treated
unit and $s^-$ the value for an untreated unit.
Each occurs with probability one half.

Let the base rule treat everyone exactly when its bag
contains $s^+$.
Choose a training sample containing one occurrence
of $s^+$, at index $i$, and $s^-$ at every other index.
Replace the observation at $i$ by $s^-$.
In the original sample, a bag treats everyone exactly
when it contains $i$.
After replacement, every bag treats no one.
Thus the complete treatment probabilities are $m/n$
and zero, respectively.

In this population, $W(G)=-1/2+G(0)$.
The welfare difference therefore equals the difference
in treatment probabilities:
\[
\left|
W\{A_m(S)\}
-
W\{A_m(S^{i\leftarrow s^-})\}
\right|
=
\frac{m}{n}
=
\frac{m}{n}\|\tau\|_{L_1(P_X)}.
\]
We attain the upper bound from Step~2, proving
sharpness.
\end{enumerate}
\end{proof}

\begin{proof}[Proof of \Cref{lem:finite_stability}]
For a realized bag collection $\mathbf I=(I_1,\ldots,I_B)$, let
$
K_i(\mathbf I)=\sum_{b=1}^B\1(i\in I_b).
$ 
\begin{enumerate}
\item We bound the change in treatment probabilities
for a fixed collection of bags.
Fix $\mathbf I$, $S$, $i$, and $s'$.
Replacing observation $i$ leaves every bag that omits
$i$ unchanged.
Each bag containing $i$ can change its binary policy
by at most one at any covariate value.
Since computational bagging averages these policies,  for every $x\in\cX$,
\[
\begin{aligned}
\left|
\{A_{m,B}(S,\mathbf I)\}(x)
-
\{A_{m,B}(S^{i\leftarrow s'},\mathbf I)\}(x)
\right|
&\leq
\frac1B\sum_{b=1}^B\mathbf 1\{i\in I_b\}
=
\frac{K_i(\mathbf I)}{B}.
\end{aligned}
\]
The same bag collection is used before and after
the replacement.

\item We convert the algorithmic stability bound into a welfare stability bound.
The bound in Step~1 holds uniformly over $S,s',x$.
Maximizing over $i$ therefore bounds the algorithmic stability parameter 
by $\max_i K_i(\mathbf I)/B$.
Applying \Cref{lem:policy_to_welfare} with $\mathbf I$
held fixed gives
\[
\begin{aligned}
\sup_{S,i,s'}
\left|
W\{A_{m,B}(S,\mathbf I)\}
-
W\{A_{m,B}(S^{i\leftarrow s'},\mathbf I)\}
\right|
\leq
\|\tau\|_{L_1(P_X)}
\max_{1\leq i\leq n}\frac{K_i(\mathbf I)}{B}.
\end{aligned}
\]

\item We bound the inclusion frequency of each observation.
The $B$ bags are sampled independently and uniformly
from the size-$m$ bags.
For each fixed $i$, the indicators
$\mathbf 1\{i\in I_b\}$ are therefore independent across
$b$, each with probability $m/n$ of equaling one.
Thus $K_i(\mathbf I)\sim\operatorname{Binomial}(B,m/n)$.
Hoeffding's inequality bounds the probability that
an average exceeds its mean: for $t>0$,
\[
\Pp_{\mathbf I}\left\{
\frac{K_i(\mathbf I)}{B}>
\frac{m}{n}+t
\right\}
\leq
\exp(-2Bt^2).
\]

\item We make the inclusion bound hold for all
observations at once.
The largest inclusion frequency exceeds $m/n+t$
exactly when at least one observation's frequency
exceeds that threshold.
The union bound therefore gives
\[
\begin{aligned}
\Pp_{\mathbf I}\left\{
\max_{1\leq i\leq n}\frac{K_i(\mathbf I)}{B}
>
\frac{m}{n}+t
\right\}
&\leq
\sum_{i=1}^n
\Pp_{\mathbf I}\left\{
\frac{K_i(\mathbf I)}{B}>
\frac{m}{n}+t
\right\}
\leq
n\exp(-2Bt^2).
\end{aligned}
\]
For $\eta\in(0,1)$, choose
$t=\sqrt{\log(n/\eta)/(2B)}$.
Then $n\exp(-2Bt^2)=n\exp\{-\log(n/\eta)\}=\eta$.
Consequently, with probability at least $1-\eta$,
\[
\max_{1\leq i\leq n}\frac{K_i(\mathbf I)}{B}
\leq
\frac{m}{n}
+
\sqrt{\frac{\log(n/\eta)}{2B}}.
\]
Substituting this bound into Step~2 gives the
claimed high-probability statement.
\end{enumerate}
\end{proof}

\subsection{Upper bounds for policy-vote bagging}

\begin{lemma}[$U$-statistic variance contraction; Corollary~3.2(i) of \citet{shao1999}]
\label[lemma]{lem:ustat}
Let $S_1,\ldots,S_n$ be i.i.d., let $1\leq m\leq n$,
and let $H_m$ be a symmetric kernel satisfying
$\E\{H_m(S_1,\ldots,S_m)^2\}<\infty$.
Let $\cI_{n,m}$ denote the collection of size-$m$
subsets of $\{1,\ldots,n\}$, and define
$
U_{n,m}
=
\binom{n}{m}^{-1}
\sum_{I\in\cI_{n,m}}H_m(S_I).
$
Then
$
\V(U_{n,m})
\leq
\frac{m}{n}\V\{H_m(S_1,\ldots,S_m)\}.
$
\end{lemma}

\begin{proof}[Proof of \Cref{thm:bag_upper}]
We proceed in steps.
Recall that
$\Delta_I(X)=\widehat\tau_I(X)-\tau(X)$ and
$Q_{m,2}=\E_{S,I,X}\{\Delta_I(X)^2\}$.
Expectations over $X$ hold the training sample
and the selected bag fixed.
\begin{enumerate}
\item We relate each bag's welfare loss to its
treatment-effect estimation error.
The oracle and the bag-specific policy are
$G^\star(x)=\1\{\tau(x)\geq0\}$ and
$\widehat G_I(x)=\1\{\widehat\tau_I(x)\geq0\}$.
Their welfare difference is
\[
W^\star-W(\widehat G_I)
=
\E_X\left[
|\tau(X)|
\1\{\widehat G_I(X)\neq G^\star(X)\}
\right].
\]
Thus a treatment mistake costs $|\tau(X)|$ units
of conditional expected welfare.

When the two policies disagree, the estimation error
must be large enough to reach or cross the treatment
threshold at zero.
Consequently, $|\tau(X)|\leq|\Delta_I(X)|$ on the
disagreement event.
Points where $\tau(X)=0$ contribute no welfare loss.
We use this observation to bound both expected regret
and welfare variance.

\item We bound expected regret by separating small
treatment effects from large estimation errors.
Because welfare is affine in the policy, complete
averaging gives
\[
R(A_m)
=
\E_{S,I}\{W^\star-W(\widehat G_I)\}.
\]
Fix $t>0$.
When $0<|\tau(X)|\leq t$, the welfare loss is at most $t$.
When $|\tau(X)|>t$ and the policies disagree,
Step~1 implies
$|\Delta_I(X)|\geq|\tau(X)|>t$,
so the loss is at most $|\Delta_I(X)|$.
Therefore,
\[
\begin{aligned}
R(A_m)
&\leq
t\Pp_X\{0<|\tau(X)|\leq t\}
+
\E_{S,I,X}\left[
|\Delta_I(X)|\1\{|\Delta_I(X)|>t\}
\right]
\leq
C_\tau t^{\kappa+1}
+
\frac{Q_{m,2}}{t}.
\end{aligned}
\]
The last inequality uses the margin condition for
the first term and
$
|\Delta_I(X)|\1\{|\Delta_I(X)|>t\}
\leq
\frac{\Delta_I(X)^2}{t}
$
for the second.

If $Q_{m,2}>0$, choose
$t=Q_{m,2}^{1/(\kappa+2)}$.
Both terms then have the same power of $Q_{m,2}$,
giving
\[
R(A_m)
\leq
(C_\tau+1)
Q_{m,2}^{(\kappa+1)/(\kappa+2)}.
\]
If $Q_{m,2}=0$, letting $t\downarrow0$ in the preceding
bound gives $R(A_m)=0$.

The moment extension in \Cref{rmk:p_moment} follows
from the same calculation.
Using
$|\Delta_I(X)|\1\{|\Delta_I(X)|>t\}
\leq t^{1-p}|\Delta_I(X)|^p$
replaces the second term by
$t^{1-p}\E_{S,I,X}|\Delta_I(X)|^p$.
Balancing the two terms gives the stated extension.

\item We bound the variance of a single bag's welfare
through its squared welfare loss.
For each fixed sample and bag, Jensen's inequality
gives
\[
\begin{aligned}
\{W^\star-W(\widehat G_I)\}^2
&=
\left(
\E_X\left[
|\tau(X)|
\1\{\widehat G_I(X)\neq G^\star(X)\}
\right]
\right)^2\\
&\leq
\E_X\left[
\tau(X)^2
\1\{\widehat G_I(X)\neq G^\star(X)\}
\right]\\
&\leq
\E_X\{\Delta_I(X)^2\}.
\end{aligned}
\]
The second inequality uses the estimation-error bound from Step~1:
on the disagreement event,
$\tau(X)^2\leq\Delta_I(X)^2$.
Taking expectations over $S$ and $I$ therefore gives
\[
\E_{S,I}\left[
\{W^\star-W(\widehat G_I)\}^2
\right]
\leq
\E_{S,I,X}\{\Delta_I(X)^2\}
=
Q_{m,2}.
\]

This expected squared loss is at least
the welfare variance.
Adding and subtracting,
\[
\begin{aligned}
W^\star-W(\widehat G_I)
={}&
\left[W^\star-\E_{S,I}\{W(\widehat G_I)\}\right]
-
\left[W(\widehat G_I)-\E_{S,I}\{W(\widehat G_I)\}\right].
\end{aligned}
\]
Expanding the square and taking expectations yields
\[
\begin{aligned}
&\E_{S,I}\left[
\{W^\star-W(\widehat G_I)\}^2
\right]=
\left[W^\star-\E_{S,I}\{W(\widehat G_I)\}\right]^2\\
&\quad+
\E_{S,I}\left(
\left[
W(\widehat G_I)-\E_{S,I}\{W(\widehat G_I)\}
\right]^2
\right)\\
&\quad-
2\left[W^\star-\E_{S,I}\{W(\widehat G_I)\}\right]
\E_{S,I}\left[
W(\widehat G_I)-\E_{S,I}\{W(\widehat G_I)\}
\right].
\end{aligned}
\]
The first term is nonnegative. The second is $\V\{W(\widehat G_I)\}$.
The third is zero.
Consequently,
$
\V\{W(\widehat G_I)\}
\leq
\E_{S,I}\left[
\{W^\star-W(\widehat G_I)\}^2
\right]
\leq Q_{m,2}.
$

\item We apply variance contraction to complete bagging.
The base learner is symmetric in its bag observations,
and welfare is affine in the policy.
Thus
\[
W\{A_m(S)\}
=
\binom{n}{m}^{-1}
\sum_{I\in\cI_{n,m}}W(\widehat G_I)
\]
is a complete $U$-statistic with symmetric kernel
given by the welfare of a bag-specific policy.
Step~3 establishes that this kernel is square-integrable.

A uniformly selected size-$m$ bag has the same
distribution as $m$ independent training observations.
Its variance is therefore the kernel variance appearing
in \Cref{lem:ustat}.
Applying that lemma and then Step~3 gives
\[
\V[W\{A_m(S)\}]
\leq
\frac{m}{n}\V\{W(\widehat G_I)\}
\leq
\frac{m}{n}Q_{m,2}.
\]
Finally, combining this bound with Step~2 and the
definition of risk-aware regret yields
\[
\begin{aligned}
R_{\mathrm{RA}}(A_m)
&=
R(A_m)+\rho\V[W\{A_m(S)\}]
\leq
(C_\tau+1)
Q_{m,2}^{(\kappa+1)/(\kappa+2)}
+
\rho\frac{m}{n}Q_{m,2}.
\end{aligned}
\]
\end{enumerate}
\end{proof}

\begin{proof}[Proof of \Cref{cor:finite_bagging}]
The conditional-mean equality in \Cref{prop:bag_concave}
implies $R(A_{m,B})=R(A_m)$ after taking expectations
over $S$. It remains to bound the variance.
\begin{enumerate}
\item We separate sampling variance from the
variance due to finite bag draws.
Conditional on $S$, finite-bag welfare averages $B$
independent bag-specific welfare values. Hence,
\[
\V_{\mathbf I}\left[
W\{A_{m,B}(S,\mathbf I)\}\mid S
\right]
=
\frac1B\V_I\{W(\widehat G_I)\mid S\}.
\]
By \Cref{prop:bag_concave}, its conditional mean is
$W\{A_m(S)\}$. The law of total variance gives
\[
\begin{aligned}
\V_{S,\mathbf I}[W\{A_{m,B}(S,\mathbf I)\}]
={}&
\V_S[W\{A_m(S)\}]
+\frac1B\E_S[\V_I\{W(\widehat G_I)\mid S\}].
\end{aligned}
\]

\item We express the additional variance using the
variance of a single bag.
A second application of the law of total variance gives
\[
\V_{S,I}\{W(\widehat G_I)\}
=
\V_S[W\{A_m(S)\}]
+
\E_S[\V_I\{W(\widehat G_I)\mid S\}],
\]
because single-bag welfare also has conditional mean
$W\{A_m(S)\}$.
Substituting $\E_S[\V_I\{W(\widehat G_I)\mid S\}]$ into Step~1 gives
\[
\begin{aligned}
\V_{S,\mathbf I}[W\{A_{m,B}(S,\mathbf I)\}]
={}&
\left(1-\frac1B\right)\V_S[W\{A_m(S)\}]
+\frac1B\V_{S,I}\{W(\widehat G_I)\}.
\end{aligned}
\]

\item We apply the bounds established for complete
bagging and a single bag.
By \Cref{thm:bag_upper} and its proof,
$
\V_S[W\{A_m(S)\}]
\leq \frac mn Q_{m,2}$ and $
\V_{S,I}\{W(\widehat G_I)\}
\leq Q_{m,2}.
$
Since $B\geq1$, both coefficients in Step~2 are
nonnegative. Substitution therefore yields
\[
\begin{aligned}
\V_{S,\mathbf I}[W\{A_{m,B}(S,\mathbf I)\}]
&\leq
\left\{
\left(1-\frac1B\right)\frac mn+\frac1B
\right\}Q_{m,2}
=
\left\{
\frac mn+\frac{1-m/n}{B}
\right\}Q_{m,2}.
\end{aligned}
\]
Using $R(A_{m,B})=R(A_m)$, we conclude that
\[
R_{\mathrm{RA}}(A_{m,B})
\leq
R(A_m)
+
\rho\left\{
\frac mn+\frac{1-m/n}{B}
\right\}Q_{m,2}.
\]
Applying the expected-regret bound in
\Cref{thm:bag_upper} completes the proof.
\end{enumerate}
\end{proof}

\begin{proof}[Proof of \Cref{cor:cara_bagging}]
Fix $\gamma>0$. If $T=0$, all welfare losses and variances
vanish by \Cref{prop:cara_variance}. Henceforth, suppose
$T>0$. We proceed in steps.
\begin{enumerate}
\item We apply the CARA variance bound to complete bagging.
By \Cref{prop:cara_variance},
\[
R(A_m)
\leq
W^\star-\CE_\gamma(A_m)
\leq
R(A_m)
+
\frac{e^{\gamma T}-1-\gamma T}{\gamma T^2}
\V_S[W\{A_m(S)\}].
\]
The coefficient of the variance is nonnegative because
$e^{\gamma T}\geq1+\gamma T$.
We may therefore substitute the variance bound from
\Cref{thm:bag_upper} to obtain
\[
R(A_m)
\leq
W^\star-\CE_\gamma(A_m)
\leq
R(A_m)
+
\frac{e^{\gamma T}-1-\gamma T}{\gamma T^2}
\frac mn Q_{m,2}.
\]

\item We apply the same argument to computational bagging.
Here the certainty equivalent averages over both the
experimental sample $S$ and the random bag collection
$\mathbf I$.
For every realization of these two sources of randomness,
the implemented policy still takes values in $[0,1]$.
Thus \Cref{prop:cara_variance} applies to their joint law
and gives
\[
\begin{aligned}
R(A_{m,B})
&\leq W^\star-\CE_\gamma(A_{m,B})
\leq R(A_{m,B})
+
\frac{e^{\gamma T}-1-\gamma T}{\gamma T^2}
\V_{S,\mathbf I}[W\{A_{m,B}(S,\mathbf I)\}].
\end{aligned}
\]
By \Cref{cor:finite_bagging},
$R(A_{m,B})=R(A_m)$ and 
$\V_{S,\mathbf I}[W\{A_{m,B}(S,\mathbf I)\}]\leq \{m/n+(1-m/n)/B\}Q_{m,2}$.
Substituting both conclusions yields
\[
\begin{aligned}
R(A_m)
&\leq W^\star-\CE_\gamma(A_{m,B})
\leq R(A_m)
+
\frac{e^{\gamma T}-1-\gamma T}{\gamma T^2}
\left\{
\frac mn+\frac{1-m/n}{B}
\right\}Q_{m,2}.
\end{aligned}
\]
The expected-regret bound in \Cref{thm:bag_upper}
then gives the stated rates for both procedures.
\end{enumerate}
\end{proof}
\section{Details and proofs for Section 6}\label{app:sec6_proofs}

We give the conditions, constructions, and proofs supporting Section~\ref{sec:rates}. We begin with high-level sharpness. Then we turn to optimal bag size, with details for the H\"older and CARA examples. Finally, we provide the local minimax experiment and details for the Gaussian-RKHS comparison.

\subsection{Sharpness over a high-level class}\label{app:sharpness_proof}

The sharpness construction makes both terms of the upper bound necessary over the high-level class. One covariate stratum determines expected regret, while another generates welfare variation across samples.

\begin{figure}[h]
    \centering
    \includegraphics[width=\textwidth]{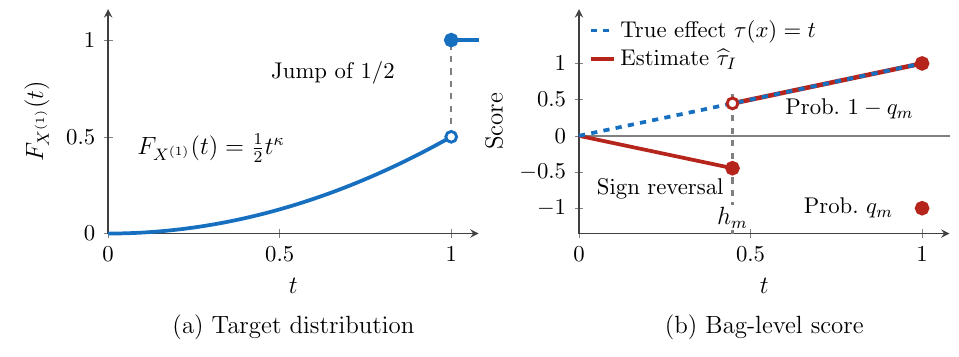}
\caption{Sharpness construction}\label{fig:sharpness-construction}

\caption*{\footnotesize \textit{Notes:}
Panel (a) shows the treatment-effect distribution: half the
population has effects in $(0,1)$, and half has effect one.
In panel (b), the estimate reverses the sign for
$0<t\leq h_m$, creating systematic treatment mistakes and
expected regret. Across samples, a uniformly selected bag's
estimate at $t=1$ is $-1$ with probability $q_m$ and $+1$
otherwise. These errors make welfare vary even after policy
votes are averaged across bags.
Here $\kappa=2$, $\alpha=1$, $m=25$,
$h_m=m^{-\alpha/(\kappa+2)}\approx0.447$, and
$q_m=\frac12m^{-\alpha}=0.02$.}
\end{figure}

\begin{proof}[Proof of \Cref{thm:sharpness}]
Fix $\kappa>0$ and $\alpha>0$.
We construct one population and a sequence of symmetric
base learners. Treatment mistakes among individuals with
small treatment effects create expected regret.
Sample-dependent mistakes among individuals with treatment
effect one create welfare variation.
We proceed in steps.
\begin{enumerate}

\item We construct a population satisfying the margin condition.
Let $X\in[0,1]^3$ have independent coordinates.
Its first coordinate has density $\kappa t^{\kappa-1}/2$
on $(0,1)$ and equals one with probability one half.
Its other two coordinates are uniform on $(0,1)$.
Writing $x=(t,u,v)$, set
$Y(1)=X^{(1)}/2$ and $Y(0)=-X^{(1)}/2$.
Thus $\tau(x)=t$, and both potential outcomes are bounded
in absolute value by one half.
In the training experiment, independently assign treatment
with probability one half.

The margin condition holds with constant one.
For $0<t<1$,
$
\Pp\{0<|\tau(X)|\leq t\}
=
\frac12 t^\kappa
\leq t^\kappa.
$
For $t\geq1$, the probability is one and is again at most
$t^\kappa$.
Since $\tau(X)>0$ almost surely, the oracle treats everyone.

\item We define a learner with two distinct sources of error.
For each bag size $m$, let
\[
h_m=m^{-\alpha/(\kappa+2)},
\qquad
q_m=\frac12h_m^{\kappa+2}
=\frac12m^{-\alpha}.
\]
For this bag size, define the indicator
$B_j=\1\{X_j^{(3)}\leq q_m\}$ for each observation $j$.
For a bag $I$, let $J_I$ identify the observation with
the smallest second coordinate.
The second coordinate therefore determines which observation
is selected. Its third coordinate determines whether the
learner makes a treatment mistake at $t=1$.

Define
\[
\widehat\tau_I(t,u,v)
=
\begin{cases}
-t, & t<1 \text{ and } t\leq h_m,\\
 t, & t<1 \text{ and } t>h_m,\\
1-2B_{J_I}, & t=1,
\end{cases}
\qquad
\widehat G_I(x)=\1\{\widehat\tau_I(x)\geq0\}.
\]
For $0<t<1$, every bag makes the same treatment mistakes
when $t\leq h_m$.
At $t=1$, the learner treats correctly if $B_{J_I}=0$:
its estimated treatment effect is then $+1$.
If $B_{J_I}=1$, the estimate is $-1$, so the learner
incorrectly withholds treatment.
Once the sample is fixed, these decisions are determined.

The learner is symmetric because selection depends on
observed coordinates rather than observation labels.
Ties have probability zero and can be resolved using a fixed
ordering of the observed tuples; identical tuples have the
same indicator.
Figure~\ref{fig:sharpness-construction} illustrates
the construction.

\item We verify the required mean-squared-error rate.
For $0<t<1$, the estimation error is $-2t$ when
$t\leq h_m$ and zero otherwise.
At $t=1$, the error is $-2B_{J_I}$.

The selected indicator equals one with probability $q_m$.
To see why, condition on all second coordinates in the bag.
These coordinates determine $J_I$.
The third coordinates remain independent and uniform,
so the selected observation's third coordinate is at most
$q_m$ with probability $q_m$.
Consequently, the expected squared error at $t=1$ is $4q_m$.

Averaging the squared error over the training bag and
an independent population observation gives
\[
\begin{aligned}
Q_{m,2}
&=
\frac12\int_0^{h_m}4t^2\kappa t^{\kappa-1}\,\dd t
+\frac12\cdot4q_m
=
\frac{2\kappa}{\kappa+2}h_m^{\kappa+2}
+2q_m
=
\left(1+\frac{2\kappa}{\kappa+2}\right)m^{-\alpha}.
\end{aligned}
\]
Thus $Q_{m,2}\asymp m^{-\alpha}$, with constants
independent of $n$ and $m$.

\item We obtain the regret lower bound from individuals
with small treatment effects.
Every bag withholds treatment when $0<t<1$ and $t\leq h_m$.
Averaging the bags' decisions also withholds
treatment from these individuals, even though the oracle treats.
The welfare loss at covariate value $t$ is exactly $t$.
Ignoring the nonnegative loss at $t=1$ gives
\[
\begin{aligned}
R(A_m)
&\geq
\frac12\int_0^{h_m}t\kappa t^{\kappa-1}\,\dd t
=
\frac{\kappa}{2(\kappa+1)}h_m^{\kappa+1}
=
\frac{\kappa}{2(\kappa+1)}
m^{-\alpha(\kappa+1)/(\kappa+2)}.
\end{aligned}
\]

\item We isolate the part of welfare that varies across samples.
The loss among individuals with $0<t<1$ is the same
for every bag and every training sample.
At $t=1$, the bag-specific treatment decision is $1-B_{J_I}$.
This group has population share one half and treatment
effect one, so its contribution to welfare loss is $B_{J_I}/2$.

Averaging the bag-specific treatment probabilities also
averages their welfare values, so
\[
W\{A_m(S)\}
=
W^\star
-
\frac{\kappa}{2(\kappa+1)}h_m^{\kappa+1}
-
\frac12\binom{n}{m}^{-1}
\sum_{I\in\cI_{n,m}}B_{J_I}.
\]
The first two terms are deterministic.
The complete average of $B_{J_I}$ is the fraction of bags
that incorrectly withhold treatment at $t=1$.
Welfare variance across training samples is therefore
one quarter of the variance of this fraction.

\item We calculate how one observation affects the probability
of a treatment mistake at $t=1$.
Fix a size-$m$ bag containing observation one.
Conditional on $X_1^{(2)}=u$ and $B_1=b$, observation one
is selected exactly when all other second coordinates
exceed $u$.
This event has probability $(1-u)^{m-1}$.

If observation one is selected, $B_{J_I}=b$.
Otherwise, another observation is selected using the
second coordinates. Its third coordinate remains independent
of that selection, so its indicator has expectation $q_m$.
Thus,
\[
\begin{aligned}
\E(B_{J_I}\mid X_1^{(2)}=u,B_1=b)
&=
b(1-u)^{m-1}
+q_m\{1-(1-u)^{m-1}\}\\
&=
q_m+(b-q_m)(1-u)^{m-1}.
\end{aligned}
\]
The decision at $t=1$ depends on $S_1$ only through
$B_1$ and $X_1^{(2)}$. Therefore,
\[
\E(B_{J_I}\mid S_1)-q_m
=
(B_1-q_m)(1-X_1^{(2)})^{m-1}.
\]
This is the change in the expected indicator from learning
the value of observation one.

The displayed expression has mean zero.
Since $B_1$ and $X_1^{(2)}$ are independent and
$\V(B_1)=q_m(1-q_m)$, its variance is
\[
\begin{aligned}
\V\{\E(B_{J_I}\mid S_1)\}
&=
\E\{(B_1-q_m)^2\}
\E\{(1-X_1^{(2)})^{2m-2}\}\\
&=
q_m(1-q_m)
\int_0^1(1-u)^{2m-2}\,\dd u\\
&=
\frac{q_m(1-q_m)}{2m-1}.
\end{aligned}
\]

\item We use these single-observation contributions to
lower-bound welfare variance.
The Hoeffding decomposition separates the centered complete
average of $B_{J_I}$ into contributions from individual
observations and additional terms involving several
observations.
The single-observation part is
\[
\frac mn\sum_{i=1}^n
(B_i-q_m)(1-X_i^{(2)})^{m-1}.
\]
The factor $m/n$ arises because each observation appears
in that fraction of all bags.
The $n$ summands are independent and have the variance
calculated in Step~6.
The variance of the displayed term is therefore
\[
\left(\frac mn\right)^2
n\,\frac{q_m(1-q_m)}{2m-1}
=
\frac{m^2}{n}\frac{q_m(1-q_m)}{2m-1}.
\]
The remaining Hoeffding terms have zero covariance with
this term. Adding them cannot reduce variance.

Step~5 shows that welfare variance is $1/4$ of
the variance of the complete average, so
\[
\V[W\{A_m(S)\}]
\geq
\frac{m^2}{4n}
\frac{q_m(1-q_m)}{2m-1}.
\]
Since $q_m\leq1/2$ and $2m-1\leq2m$,
\[
\begin{aligned}
\V[W\{A_m(S)\}]
&\geq
\frac{m^2}{4n}\frac{q_m/2}{2m}
=
\frac{mq_m}{16n}
=
\frac{m^{1-\alpha}}{32n}.
\end{aligned}
\]
Combining this with Step~4 gives, for every $\rho\geq0$,
\[
R_{\mathrm{RA}}(A_m)
\geq
\frac{\kappa}{2(\kappa+1)}
m^{-\alpha(\kappa+1)/(\kappa+2)}
+
\frac{\rho}{32n}m^{1-\alpha}.
\]
We have verified \eqref{eq:sharp_lower} for every pair of
integers $n\geq m\geq1$, with constants independent of
$n$, $m$, and $\rho$.
The construction has margin constant one and
mean-squared-error constant $1+2\kappa/(\kappa+2)$,
so it satisfies the uniform class restrictions whenever
their fixed bounds admit these values.

\end{enumerate}
\end{proof}

\subsection{Optimal bag choice}\label{app:bag_choice_proof}

\subsubsection{\Cref{cor:bag_choice}}

Matching upper and lower bounds reduce the bag-size comparison to minimizing the stated envelope. We specify the common class restrictions and the conditions for an interior choice.

\paragraph{Class for the bag-size comparison.}\label{app:bag_class}
The class holds both the statistical conditions and the learning guarantees fixed across bag sizes. Fix $\kappa>0$ and $\alpha>0$, and let $\mathcal H_{\alpha,\kappa}$ consist of pairs of a data-generating process and a sequence of base learners satisfying Assumption~\ref{ass:main}. For fixed constants $C_\tau,C_Q<\infty$, every pair has margin constant at most $C_\tau$ and satisfies $Q_{m,2}\leq C_Qm^{-\alpha}$ for every integer $m\geq1$. The class contains the construction of \Cref{thm:sharpness}. Write $R_{\RA,H}(A_m)$ for the risk-aware regret of complete policy-vote bagging under the population and base learners specified by $H$. All comparison constants may depend on $\alpha,\kappa,C_\tau,C_Q$, but are independent of $n,m,\rho$.

\paragraph{Interior bag sizes.}
The interior regime requires sampling-risk aversion to lie between two sample-size-dependent thresholds. For $0<\alpha<1$ and $\rho>0$, the range is $n^{\alpha/(\kappa+2)}\lesssim\rho\lesssim n$, up to constant factors. With $m_0$ as defined in \Cref{cor:bag_choice}, if $m_0\to\infty$ and $m_0/n\to0$, the condition $m\asymp m_0$ is necessary and sufficient for constant-factor minimax-rate optimality. The following proof establishes these additional details along with the corollary.

\begin{proof}[Proof of \Cref{cor:bag_choice}]
We proceed in steps.
\begin{enumerate}
\item We reduce the worst-case comparison to minimizing
the rate envelope.
The upper bound in \Cref{thm:bag_upper} holds uniformly
over the class because its members share the constants
$C_\tau$ and $C_Q$.
The construction in \Cref{thm:sharpness} belongs to the
class and supplies the matching lower bound.
Thus, uniformly over integers $n\geq m\geq1$ and
$\rho\geq0$,
\[
\sup_{H\in\mathcal H_{\alpha,\kappa}}
R_{\RA,H}(A_m)
\asymp
m^{-\frac{\alpha(\kappa+1)}{\kappa+2}}
+
\frac{\rho}{n}m^{1-\alpha}.
\]
The comparison constants depend only on
$\alpha,\kappa,C_\tau,C_Q$.
The same constants apply at every bag size, so taking
infima over integer $m$ preserves the comparison.

\item We identify the cases in which larger bags improve
both terms.
The expected-regret term
$m^{-\alpha(\kappa+1)/(\kappa+2)}$ decreases with $m$.
If $\alpha\geq1$, the variance term
$(\rho/n)m^{1-\alpha}$ is also nonincreasing.
The envelope is therefore minimized at $m=n$,
proving part~(i).
The same conclusion holds for every $\alpha>0$
when $\rho=0$, since the variance term vanishes.

\item We find the optimal order when the two terms move
in opposite directions.
Suppose $0<\alpha<1$ and $\rho>0$.
For the calculation, allow $m$ to vary over positive
real numbers.
Differentiating and factoring out a positive power of $m$
gives
\[
\begin{aligned}
&\frac{\dd}{\dd m}
\left\{
m^{-\frac{\alpha(\kappa+1)}{\kappa+2}}
+\frac{\rho}{n}m^{1-\alpha}
\right\}
=
m^{-\frac{\alpha(\kappa+1)}{\kappa+2}-1}
\left\{
-\frac{\alpha(\kappa+1)}{\kappa+2}
+\frac{\rho}{n}(1-\alpha)
m^{\frac{\kappa+2-\alpha}{\kappa+2}}
\right\}.
\end{aligned}
\]
The expression in braces is strictly increasing in $m$
and crosses zero once.
Solving for that zero gives the unique stationary point
\[
m=
\left\{
\frac{\alpha(\kappa+1)}
{(1-\alpha)(\kappa+2)}
\frac{n}{\rho}
\right\}^{\frac{\kappa+2}{\kappa+2-\alpha}}.
\]
The derivative is negative before this point and positive
after it, so this point minimizes the unconstrained envelope.

The constraint $1\leq m\leq n$ determines whether the
minimum occurs at an endpoint.
If the stationary point is below one, the constrained
minimum is at one; if it exceeds $n$, the minimum is at $n$.
Otherwise, it remains the continuous minimizer.
Comparing the neighboring admissible integers gives
an exact integer minimizer.

Only a fixed coefficient separates the stationary point
from
$
m_0=
\left(\frac n\rho\right)^{
\frac{\kappa+2}{\kappa+2-\alpha}}.
$
Restricting both choices to $[1,n]$ preserves their
comparability, and integer rounding changes their order
by at most a fixed factor.
Since both terms of the envelope are powers of $m$,
these changes also affect the objective by at most a
fixed factor.
Hence any integer choice satisfying
$m\asymp\min\{n,\max(1,m_0)\}$
is minimax-rate optimal.

\item We characterize when an interior bag size is
necessary for the optimal rate.
The endpoint comparisons for $m_0$ are
$
m_0\leq n
$ when $
\rho\geq n^{\alpha/(\kappa+2)}$ and $
m_0\geq1
$ when $
\rho\leq n.
$
Because the stationary point differs from $m_0$ by a
fixed factor, these thresholds give the interior range
$n^{\alpha/(\kappa+2)}\lesssim\rho\lesssim n$,
up to fixed constant factors.

Now consider an interior regime with
$m_0\to\infty$ and $m_0/n\to0$.
The definition of $m_0$ implies
$\rho/n=m_0^{-(\kappa+2-\alpha)/(\kappa+2)}$.
Substitution rewrites the envelope as
\[
\begin{aligned}
&m^{-\frac{\alpha(\kappa+1)}{\kappa+2}}
+\frac{\rho}{n}m^{1-\alpha}
=
m_0^{-\frac{\alpha(\kappa+1)}{\kappa+2}}
\left\{
\left(\frac m{m_0}\right)^{
-\frac{\alpha(\kappa+1)}{\kappa+2}}
+
\left(\frac m{m_0}\right)^{1-\alpha}
\right\}.
\end{aligned}
\]
The factor in braces is at least one:
its first term is at least one when $m\leq m_0$,
and its second term is at least one when $m\geq m_0$.
Choosing an integer near $m_0$ keeps both terms bounded.
The minimized envelope therefore has order
$m_0^{-\alpha(\kappa+1)/(\kappa+2)}$.

Both terms remain bounded precisely when
$m/m_0$ stays bounded above and away from zero.
The first term rules out arbitrarily small ratios,
and the second rules out arbitrarily large ratios.
Thus $m\asymp m_0$ is necessary and sufficient for
constant-factor optimality in this interior regime.
The uniform comparison in Step~1 transfers this conclusion
to the worst-case criterion, proving part~(ii).

\item We establish the optimal rate for fixed sampling-risk
aversion.
Fix $\rho\geq0$.
At $m=n$, 
\[
n^{-\frac{\alpha(\kappa+1)}{\kappa+2}}
+\rho n^{-\alpha}
=
n^{-\frac{\alpha(\kappa+1)}{\kappa+2}}
\left\{1+\rho n^{-\frac{\alpha}{\kappa+2}}\right\}
\asymp
n^{-\frac{\alpha(\kappa+1)}{\kappa+2}}.
\]
No admissible bag size can attain a smaller order,
because $m\leq n$ implies
$
m^{-\frac{\alpha(\kappa+1)}{\kappa+2}}
\geq
n^{-\frac{\alpha(\kappa+1)}{\kappa+2}},
$
and the variance term is nonnegative.
Every sequence $m\asymp n$ attains the same order in
both terms as the full sample.
Step~1 therefore proves part~(iii).

For fixed $\rho>0$ and $0<\alpha<1$, the stationary
point in Step~3 grows faster than $n$, because its
exponent on $n$ is
$(\kappa+2)/(\kappa+2-\alpha)>1$.
The envelope is consequently minimized at $m=n$
for all sufficiently large $n$.
This exact minimization statement concerns the envelope;
the rate comparison does not identify the exact minimizer
of the finite-sample worst-case criterion.
\end{enumerate}
\end{proof}

\subsubsection{\Cref{cor:holder}}\label{app:holder_conditions}

The H\"older example connects smoothness and dimension
to the choice of bag size. We obtain the CATE-error
rate, then substitute it into the welfare bound and
the bag-size calculation.

\paragraph{Model and learner.}
The regression conditions specify the learning rate.
Suppose Assumption~\ref{ass:main} holds with $\kappa>0$.
Let $s>0$, let $\cX=[0,1]^d$ with $d\geq1$, and suppose
$P_X$ has a density bounded above and away from zero
on $\cX$.
Treatment is randomized with known propensity bounded
away from zero and one.
The inverse-propensity pseudo-outcome has uniformly
conditionally sub-Gaussian noise, and $\tau$ belongs
to a fixed-radius $s$-H\"older ball.
Let $\widehat\tau_I$ be a rate-optimal local-polynomial
CATE estimator trained on a size-$m$ bag, with bandwidth
of order $m^{-1/(2s+d)}$.

\begin{proof}[Proof of \Cref{cor:holder}]
We proceed in steps.
\begin{enumerate}
\item We obtain the integrated CATE-error rate.
The local-polynomial regression bound balances
approximation error against sampling variation.
At the specified bandwidth, the squared approximation
error has order $m^{-2s/(2s+d)}$, while the variance has
order
$
\frac{1}{m\{m^{-1/(2s+d)}\}^d}
=
m^{-\frac{2s}{2s+d}}.
$
The specified learner therefore satisfies
$
Q_{m,2}\lesssim m^{-\frac{2s}{2s+d}}.
$
A uniformly selected size-$m$ bag has the same distribution
as $m$ independent experimental observations.
Thus the regression bound applies to $Q_{m,2}$ after
averaging over the sample and the selected bag.

\item We substitute the learning rate into the welfare bound.
Set $\alpha=2s/(2s+d)$, so $0<\alpha<1$ and
$1-\alpha=d/(2s+d)$.
By \Cref{thm:bag_upper},
\[
\begin{aligned}
R_{\RA}(A_m)
&\lesssim
Q_{m,2}^{\frac{\kappa+1}{\kappa+2}}
+\rho\frac mn Q_{m,2}
\lesssim
m^{-\frac{2s(\kappa+1)}{(2s+d)(\kappa+2)}}
+\frac{\rho}{n}m^{\frac{d}{2s+d}}.
\end{aligned}
\]
The first term decreases with bag size.
For $\rho>0$, the second increases with bag size.

\item We minimize the resulting upper envelope.
For $\rho>0$, the calculation in the proof of
\Cref{cor:bag_choice} gives the interior order
$(n/\rho)^{(\kappa+2)/(\kappa+2-\alpha)}$.
Here,
\[
\frac{\kappa+2}{\kappa+2-\alpha}
=
\frac{(2s+d)(\kappa+2)}
{2s(\kappa+1)+d(\kappa+2)}.
\]
Consequently, an interior minimizer has order
$
m^\star\asymp
\left(\frac n\rho\right)^{
\frac{(2s+d)(\kappa+2)}
{2s(\kappa+1)+d(\kappa+2)}}.
$
Substituting $\alpha=2s/(2s+d)$ into the corresponding
range for $\rho$ gives
$
n^{\frac{2s}{(2s+d)(\kappa+2)}}
\lesssim \rho\lesssim n,
$
up to fixed constant factors.
Outside this range, restricting the choice to $[1,n]$
gives the appropriate endpoint order.
Integer rounding preserves these orders.
When $\rho=0$, the envelope decreases throughout,
so its minimum is at $m=n$.
\end{enumerate}
\end{proof}

The calculation optimizes the displayed upper bound;
it does not establish minimax optimality within the
H\"older class.
Using only the mean-squared-error rate can give a looser
regret bound than using the full 
conditions \citep[Theorem~3.3]{audibert_tsybakov_2007}.

\subsubsection{\Cref{ex:cara_bag_choice}}\label{app:cara_bag_choice}

The CARA example determines the best bag size from the
exact distribution of welfare. We specify the experiment,
calculate the bagged policies, and compare their
certainty equivalents.

\paragraph{Model and learner.}
A bounded-outcome randomized experiment supplies the
observed scores.
Let $n=3$ and let $X$ be constant.
For each experimental unit, independently draw $U$ equal
to $1$ with probability $3/4$ and $-3/2$ with probability
$1/4$.
Set $Y(1)=U/2$ and $Y(0)=-U/2$, and randomize $D$
independently of $U$ with probability one half.
\label{app:cara_experiment}
Each base learner averages the observed inverse-propensity
scores $2(2D-1)Y$ over its bag and treats when the average
is nonnegative.
Complete policy-vote bagging averages these decisions
over all bags of size $m\in\{1,2,3\}$.

\begin{proof}[Proof of \Cref{ex:cara_bag_choice}]
We proceed in steps.
\begin{enumerate}
\item We verify the score distribution and its exact
mean-squared error.
The observed outcome is $Y=(2D-1)U/2$.
Since $(2D-1)^2=1$, the observed score satisfies
$
2(2D-1)Y=U.
$
Its moments are
$
\E(U)=\frac38,
$ $
\E(U^2)=\frac{21}{16},
$ and $
\V(U)=\frac{21}{16}-\left(\frac38\right)^2
=\frac{75}{64}.
$
Thus $\tau=3/8$, and the average of $m$ independent
scores is unbiased for $\tau$, with variance
$75/(64m)$.
A uniformly selected size-$m$ bag has the same distribution
as $m$ independent observations, so
$
Q_{m,2}=\frac{75}{64m}.
$

The experiment also satisfies the outcome and margin
conditions.
Potential outcomes are bounded, and treatment is
independently randomized with propensity one half.
Because $\tau=3/8$ is constant, the margin probability
is zero below $3/8$ and one at or above $3/8$.
The margin condition therefore holds for every $\kappa>0$
with $C_\tau=(8/3)^\kappa$.

\item We calculate the policies from the number of
positive scores.
Let $K$ count the positive scores among the three
experimental units.
Then $K\sim\operatorname{Binomial}(3,3/4)$.
A size-one bag votes for treatment exactly when its
score is positive.
A size-two bag votes for treatment only when both
scores are positive, because a mixed pair has mean
$-1/4$.
The full sample votes for treatment exactly when
$K\geq2$.

Identifying each policy with its treatment
probability, complete averaging gives
\[
\begin{array}{c|cccc}
K&0&1&2&3\\ \hline
\text{Probability}&1/64&9/64&27/64&27/64\\
A_1(S)&0&1/3&2/3&1\\
A_2(S)&0&0&1/3&1\\
A_3(S)&0&0&1&1
\end{array}
\]
Since $\E\{Y(0)\}=-\tau/2$, welfare and oracle welfare
are
$
W\{A_m(S)\}=-\frac{\tau}{2}+\tau A_m(S)
$ and $
W^\star=\frac{\tau}{2}.
$ 
Thus realized regret is
$\operatorname{Reg}(A_m,S)=\tau\{1-A_m(S)\}$.

\item We exclude size two because it delivers lower welfare.
The table shows that $A_2(S)\leq A_3(S)$ for every
sample, with strict inequality when $K=2$.
That event has probability $27/64$.
Since $\tau>0$, the same comparison holds for welfare.
CARA utility is strictly increasing for every $\gamma>0$,
so size three has strictly higher expected utility
and certainty equivalent:
$
\CE_\gamma(A_2)<\CE_\gamma(A_3).
$
It therefore remains to compare sizes one and three.

\item We calculate the exponential moments needed for
the remaining comparison.
Writing welfare as $W^\star-\operatorname{Reg}(A_m,S)$
and factoring out $\exp(-\gamma W^\star)$  gives
\[
\CE_\gamma(A_m)
=
W^\star-\frac1\gamma
\log\left(
\E\left[\exp\{\gamma\operatorname{Reg}(A_m,S)\}\right]
\right).
\]
For $\gamma>0$, maximizing the certainty equivalent
is therefore equivalent to minimizing the expectation
inside the logarithm.

The table gives regret $(3-K)/8$ for size one.
Each negative score contributes $1/8$ to this regret,
and the three scores are independent.
Its exponential moment is consequently the product
of three identical factors:
\[
\E\left[\exp\{\gamma\operatorname{Reg}(A_1,S)\}\right]
=
\left\{\frac34+\frac14\exp(\gamma/8)\right\}^3.
\]
For size three, regret is $3/8$ when $K\leq1$
and zero otherwise.
Since $\Pp(K\leq1)=10/64$,
\[
\E\left[\exp\{\gamma\operatorname{Reg}(A_3,S)\}\right]
=
\frac5{32}\exp(3\gamma/8)+\frac{27}{32}.
\]

\item We determine which bag size has the smaller
exponential moment.
Subtracting the two expressions from Step~4 gives
\[
\begin{aligned}
&\E\left[\exp\{\gamma\operatorname{Reg}(A_3,S)\}\right]
-\E\left[\exp\{\gamma\operatorname{Reg}(A_1,S)\}\right]\\
&\qquad=
\frac9{64}
\left\{
\exp(3\gamma/8)-\exp(\gamma/4)
-3\exp(\gamma/8)+3
\right\}\\
&\qquad=
\frac9{64}
\{\exp(\gamma/8)-1\}
\{\exp(\gamma/4)-3\}.
\end{aligned}
\]
For $\gamma>0$, the first factor is positive.
The second is negative for $\gamma<4\log3$, zero
at equality, and positive for $\gamma>4\log3$.
Size three has the smaller exponential moment
below the threshold, while size one has the smaller
exponential moment above it.

Combining this comparison with Step~3 proves that
size three is uniquely optimal for $0<\gamma<4\log3$,
size one is uniquely optimal for $\gamma>4\log3$,
and precisely sizes one and three are optimal when
$\gamma=4\log3$.
\end{enumerate}
\end{proof}

\subsection{Minimax benchmarks over all algorithms}\label{app:local_minimax}

\subsubsection{\Cref{thm:minimax}}

The two-point experiment gives a lower bound for every
policy-learning algorithm. Its models differ only in
whether treatment helps or harms a rare covariate group,
making the correct decision difficult to learn.

\paragraph{The two-point experiment.}\label{app:local_experiment}
The two models make opposite treatment decisions optimal
on a rare group.
Fix $\kappa>0$.
For each $n$, let
$\varepsilon_n=cn^{-1/(\kappa+2)}$ and
$p_n=\varepsilon_n^\kappa$, where $0<c\leq1/4$
will be chosen sufficiently small below.
Under both $P_{+,n}$ and $P_{-,n}$, the covariate
$X\in\{0,1\}$ satisfies $\Pp(X=1)=p_n$, and treatment
is independently randomized with probability one half.
At $X=0$, set $Y(0)=0$ and $Y(1)=1$.
At $X=1$, let
$Y(0)\sim\operatorname{Bernoulli}(1/2)$ and
$Y(1)\sim\operatorname{Bernoulli}(1/2\pm\varepsilon_n)$
under $P_{\pm,n}$.
The treatment effect is therefore one at $X=0$ and
$\pm\varepsilon_n$ at $X=1$.
The training sample contains $n$ independent observations
from the selected experiment.
Write $\mathcal P_n=\{P_{+,n},P_{-,n}\}$.

No learner can reliably distinguish
the two treatment decisions on the rare group.
A sample-independent fractional policy supplies
the matching upper bound.

\begin{proof}[Proof of \Cref{thm:minimax}]
We proceed in steps.
\begin{enumerate}
\item We verify the margin condition under both models.
The absolute treatment effect is $\varepsilon_n$
at $X=1$ and one at $X=0$.
Consequently,
\[
\Pp\{0<|\tau(X)|\leq t\}
=
\begin{cases}
0, & 0<t<\varepsilon_n,\\
p_n, & \varepsilon_n\leq t<1,\\
1, & t\geq1.
\end{cases}
\]
In the middle case,
$p_n=\varepsilon_n^\kappa\leq t^\kappa$.
In the last case, $1\leq t^\kappa$.
Thus both models satisfy the margin condition
with constant one.

\item We relate regret to the treatment decision
on the rare group.
Fix any learner $A$, used under both models, and write
$\widehat G=A(S)$.
At $X=1$, the oracle treats under $P_{+,n}$
and withholds treatment under $P_{-,n}$.
The respective welfare losses are therefore
$\varepsilon_n\{1-\widehat G(1)\}$ and
$\varepsilon_n\widehat G(1)$.
Weighting by the group's probability $p_n$ and ignoring
the nonnegative regret contribution from $X=0$ gives
$
R_{+,n}(A)
\geq
p_n\varepsilon_n
\E_{+,n}\{1-\widehat G(1)\}$ and $
R_{-,n}(A)
\geq
p_n\varepsilon_n
\E_{-,n}\{\widehat G(1)\}.$
Averaging these inequalities yields
\[
\frac{R_{+,n}(A)+R_{-,n}(A)}{2}
\geq
\frac{p_n\varepsilon_n}{2}
\left[
1-\E_{+,n}\{\widehat G(1)\}
+\E_{-,n}\{\widehat G(1)\}
\right].
\]
To have small regret under both models, the learner
must therefore assign treatment more often under
$P_{+,n}$ than under $P_{-,n}$.

\item We bound how much the learner's treatment
probabilities can differ between the models.
Total variation bounds the difference in probability
of any event under the two sample laws.
It bounds the difference in expectations of
$\widehat G(1)$, since this quantity lies in $[0,1]$.

For completeness, the identity
$
\widehat G(1)
=
\int_0^1\1\{\widehat G(1)>u\}\,\dd u
$
expresses the fractional treatment probability
as an integral of indicators.
Taking expectations under each model and subtracting gives
\[
\begin{aligned}
&\left|
\E_{+,n}\{\widehat G(1)\}
-
\E_{-,n}\{\widehat G(1)\}
\right|
\leq
\int_0^1
\left|
\Pp_{+,n}\{\widehat G(1)>u\}
-
\Pp_{-,n}\{\widehat G(1)>u\}
\right|\,\dd u\\
&\qquad\leq
\TV(P_{+,n}^n,P_{-,n}^n).
\end{aligned}
\]
Substitution into Step~2 gives
\[
\frac{R_{+,n}(A)+R_{-,n}(A)}{2}
\geq
\frac{p_n\varepsilon_n}{2}
\left\{
1-\TV(P_{+,n}^n,P_{-,n}^n)
\right\}.
\]

\item We show that the two sample laws are close.
For one observed experimental unit, the laws differ
only when $X=1$ and $D=1$.
This event has probability $p_n/2$.
Conditional on it, the outcome is Bernoulli with success
probability $1/2+\varepsilon_n$ or
$1/2-\varepsilon_n$.

The KL divergence between these Bernoulli distributions
satisfies
\[
\begin{aligned}
&\KL\left\{
\operatorname{Bernoulli}(1/2+\varepsilon_n)
\Vert
\operatorname{Bernoulli}(1/2-\varepsilon_n)
\right\}
=
2\varepsilon_n
\log\left(
1+\frac{4\varepsilon_n}{1-2\varepsilon_n}
\right)\\
&\qquad\leq
\frac{8\varepsilon_n^2}{1-2\varepsilon_n}
\leq16\varepsilon_n^2.
\end{aligned}
\]
The first inequality uses $\log(1+x)\leq x$;
the second uses $\varepsilon_n\leq1/4$.
Multiplying by $p_n/2$ bounds the one-observation
divergence by $8p_n\varepsilon_n^2$.
Independence then gives
\[
\begin{aligned}
\KL(P_{+,n}^n\Vert P_{-,n}^n)
&\leq8np_n\varepsilon_n^2
=8n\varepsilon_n^{\kappa+2}
=8c^{\kappa+2}.
\end{aligned}
\]
Choose $c$ small enough that
$c^{\kappa+2}\leq1/16$.
Pinsker's inequality implies
\[
\TV(P_{+,n}^n,P_{-,n}^n)
\leq
\sqrt{\frac12
\KL(P_{+,n}^n\Vert P_{-,n}^n)}
\leq
2c^{(\kappa+2)/2}
\leq\frac12.
\]

\item We obtain the minimax lower bound.
The larger of the two expected regrets is at least
their average.
Combining Steps~3 and~4 therefore gives
\[
\sup_{P\in\mathcal P_n}R_P(A)
\geq
\frac{R_{+,n}(A)+R_{-,n}(A)}{2}
\geq
\frac{p_n\varepsilon_n}{4}.
\]
For $\rho\geq0$, the variance penalty is nonnegative,
so $R_{\RA,P}(A)\geq R_P(A)$.
The bound holds for every learner $A$.
The infimum over learners yields
$
\inf_A\sup_{P\in\mathcal P_n}R_{\RA,P}(A)
\geq
\frac{p_n\varepsilon_n}{4}.
$

\item We construct a policy attaining the same order.
Use the same policy for every training sample:
always treat at $X=0$ and treat with probability
one half at $X=1$.
This policy agrees with the oracle at $X=0$.
At $X=1$, it loses $\varepsilon_n/2$ under either model,
so its expected regret is $p_n\varepsilon_n/2$.

Its population welfare does not depend on the training
sample and therefore has zero sampling variance.
Individual treatment randomization is already averaged
into population welfare.
Thus its risk-aware regret is $p_n\varepsilon_n/2$
for every $\rho\geq0$.
Together with Step~5, this proves
\[
\frac{p_n\varepsilon_n}{4}
\leq
\inf_A\sup_{P\in\mathcal P_n}R_{\RA,P}(A)
\leq
\frac{p_n\varepsilon_n}{2}.
\]
Finally,
$p_n\varepsilon_n
=\varepsilon_n^{\kappa+1}
=c^{\kappa+1}n^{-(\kappa+1)/(\kappa+2)}$.
\end{enumerate}
\end{proof}

\subsubsection{\Cref{cor:gaussian_rkhs}}\label{app:gaussian_proof}

The Gaussian-RKHS comparison combines a regression-error
upper bound with a minimax lower bound.
We first specify the regression model and learner,
then show when the same model class contains the
two-point experiment used for the lower bound.

\paragraph{Model and learner.}\label{app:gaussian_conditions}
Common bounds on the model make the upper bound uniform
over covariate distributions.
Let $\cX\subset\mathbb R^d$ be compact, with $d\geq1$,
and let $\mathcal H_K$ be the RKHS of a normalized
Gaussian kernel with fixed bandwidth.
Suppose Assumption~\ref{ass:main} holds.
Treatment is randomized with known propensity bounded
away from zero and one, and the inverse-propensity
pseudo-outcome has uniformly conditionally sub-Gaussian
noise.
Fix $\kappa>0$ and common bounds on the noise parameter,
margin constant, overlap constant, and RKHS norm of
$\tau$ across the model class.
Let $\widehat\tau_n$ be kernel ridge regression with
regularization of order $(\log n)^d/n$ for $n\geq3$,
and let $A_n$ be its full-sample plug-in policy.
The upper-bound constants may depend on the fixed domain,
kernel, and class bounds, but not on $P_X$ or $n$.

The minimax comparison requires the class to contain
the local experiment.
For this comparison, suppose $\cX$ contains two distinct
points and the class allows arbitrary covariate
distributions on $\cX$, including the two-point
distributions constructed below, with treatment
propensity one half.
The fixed RKHS radius and margin and noise bounds must
be large enough to contain that construction.
The proof gives explicit sufficient values.

\begin{proof}[Proof of \Cref{cor:gaussian_rkhs}]
We proceed in steps.
\begin{enumerate}
\item We express CATE estimation as a regression problem.
Let $e(x)=\Pp(D=1\mid X=x)$ and define
$
\widetilde Y
=
\frac{DY}{e(X)}
-
\frac{(1-D)Y}{1-e(X)}.
$
Randomization gives
$\E(\widetilde Y\mid X)=\tau(X)$.
Thus the regression function is correctly specified
by an element of $\mathcal H_K$.
The conditional sub-Gaussian assumption also supplies
a common bound on $\V(\widetilde Y\mid X)$.
Only this variance bound is needed for the upper bound.

Kernel ridge regression estimates this conditional mean
by minimizing
$
\frac1n\sum_{i=1}^n
\{\widetilde Y_i-f(X_i)\}^2
+
\lambda\lVert f\rVert_{\mathcal H_K}^2
$
over $f\in\mathcal H_K$.
Normalization gives
$\lVert K(X,\cdot)\rVert_{\mathcal H_K}^2=K(X,X)=1$.

\item We bound the effective dimension uniformly over
the covariate distribution.
Enclose $\cX$ in a fixed cube.
Let $(\mu_j)_{j\geq1}$ be the eigenvalues of the kernel
integral operator in decreasing order, appending zeros
if the operator has finite rank.
By Theorem~5 and Remark~6 of \citet{belkin_2018},
$
\mu_j\leq C_0\exp(-c_0j^{1/d}),
$
where $C_0,c_0>0$ depend only on the fixed kernel
and cube.
They do not depend on $P_X$.
Increase $C_0$ to at least one.

The effective dimension sums the contribution of
each eigenvalue at regularization level $\lambda$:
$
\mathcal N(\lambda)
=
\sum_{j\geq1}\frac{\mu_j}{\mu_j+\lambda}.
$
For $0<\lambda\leq1$, choose an integer
$J\asymp\{\log(e/\lambda)\}^d$ large enough that
$c_0J^{1/d}\geq2\log(eC_0/\lambda)$.
Each of the first $J$ terms is at most one.
For the remaining terms, use
$\mu_j/(\mu_j+\lambda)\leq\mu_j/\lambda$.
Hence,
$
\mathcal N(\lambda)
\leq
J+\frac{C_0}{\lambda}
\sum_{j>J}\exp(-c_0j^{1/d}).
$

The remaining sum is controlled by an integral.
Because its summand decreases with $j$, the substitution
$u=x^{1/d}$ gives
\[
\begin{aligned}
\sum_{j>J}\exp(-c_0j^{1/d})
&\leq
\int_J^\infty\exp(-c_0x^{1/d})\,\dd x\\
&=
d\int_{J^{1/d}}^\infty
u^{d-1}\exp(-c_0u)\,\dd u\\
&\lesssim
J^{(d-1)/d}\exp(-c_0J^{1/d}).
\end{aligned}
\]
The last bound follows by repeated integration by parts.
The choice of $J$ makes
$\lambda^{-1}\exp(-c_0J^{1/d})$ bounded.
Since $J^{(d-1)/d}\leq J$, it follows that
$
\mathcal N(\lambda)
\lesssim
\{\log(e/\lambda)\}^d,
$ for $ 0<\lambda\leq1,
$
uniformly in $P_X$.

\item We apply the expected regression-error bound.
Theorem~7, equation~(9), of
\citet{mourtada_rosasco_2022} applies to the features
$K(X,\cdot)$.
Their extension to Hilbert spaces is stated immediately
after Assumption~1.
Step~1 verifies the conditional-mean and variance
conditions, and the feature norm is bounded by one.
The result gives
\[
\begin{aligned}
Q_{n,2}
&\leq
\left(1+\frac{1}{n\lambda}\right)^2
\inf_{f\in\mathcal H_K}
\left\{
\lVert f-\tau\rVert_{L_2(P_X)}^2
+\lambda\lVert f\rVert_{\mathcal H_K}^2
\right\}\\
&\quad+
\left(1+\frac{1}{n\lambda}\right)
\frac{\mathcal N(\lambda)}{n}
\operatorname*{ess\,sup}_{x}
\V(\widetilde Y\mid X=x).
\end{aligned}
\]
Choosing $f=\tau$ makes the approximation-error 
zero, so the infimum is at most
$\lambda\lVert\tau\rVert_{\mathcal H_K}^2$.

The prescribed regularization controls both remaining
terms.
With $\lambda\asymp(\log n)^d/n$, the factors
$1+1/(n\lambda)$ are bounded.
For $\lambda\leq1$, Step~2 gives
$\mathcal N(\lambda)\lesssim(\log n)^d$.
For $\lambda>1$, the normalization of the kernel gives
$\sum_j\mu_j=1$, and hence
$\mathcal N(\lambda)\leq1/\lambda$.
Using the common norm and noise bounds, we obtain
$
Q_{n,2}
\lesssim
\lambda+\frac{\mathcal N(\lambda)}{n}
\lesssim
\frac{(\log n)^d}{n},
$ for $n\geq3.
$
All constants are independent of $P_X$ and $n$.

\item We translate regression accuracy into the welfare bound.
Apply \Cref{thm:bag_upper} with $m=n$.
Since the bag is the full sample,
\[
\begin{aligned}
R_{\RA}(A_n)
&\lesssim
Q_{n,2}^{\frac{\kappa+1}{\kappa+2}}
+\rho Q_{n,2}
\lesssim
\left\{
\frac{(\log n)^d}{n}
\right\}^{\frac{\kappa+1}{\kappa+2}}
+
\rho\frac{(\log n)^d}{n}.
\end{aligned}
\]
This proves the upper bound uniformly over the stated
model class.

\item We place the two treatment-effect alternatives
inside one fixed RKHS ball.
Choose distinct $x_0,x_1\in\cX$, and let $\mathbf K$
be their Gram matrix.
Its diagonal entries are one, and
$0<K(x_0,x_1)<1$.
Its smallest eigenvalue is therefore
$1-K(x_0,x_1)>0$.

For either sign, define $\tau$ as the linear combination
of $K(x_0,\cdot)$ and $K(x_1,\cdot)$ with coefficient
vector
$\mathbf K^{-1}(1,\pm\varepsilon_n)^\top$.
Multiplication by $\mathbf K$ verifies that
$\tau(x_0)=1$ and $\tau(x_1)=\pm\varepsilon_n$.
Its squared RKHS norm is
\[
\begin{aligned}
\lVert\tau\rVert_{\mathcal H_K}^2
&=
(1,\pm\varepsilon_n)
\mathbf K^{-1}
(1,\pm\varepsilon_n)^\top
\leq
\frac{1+\varepsilon_n^2}{1-K(x_0,x_1)}
\leq
\frac{17}{16\{1-K(x_0,x_1)\}},
\end{aligned}
\]
for $0<\varepsilon_n\leq\frac14.$
The first inequality uses the largest eigenvalue
of $\mathbf K^{-1}$, which is
$1/\{1-K(x_0,x_1)\}$.
Thus both alternatives belong to a fixed RKHS ball
whose radius is at least
$
\frac{\sqrt{17}}
{4\{1-K(x_0,x_1)\}^{1/2}}.
$
This radius does not depend on $n$.

\item We verify that the local experiments belong to
the stated model class.
Choose
$\varepsilon_n=cn^{-1/(\kappa+2)}$ as in
\Cref{thm:minimax}, with $c>0$ small enough that
$\varepsilon_n\leq1/4$.
Place probability $p_n=\varepsilon_n^\kappa$ at $x_1$
and the remaining probability at $x_0$.
Use the potential-outcome laws from
Appendix~\ref{app:local_experiment}, replacing its
covariate labels $0,1$ by $x_0,x_1$.

The conditional treatment effects on these two points
are those in Step~5.
Off the support of $P_X$, choose the CATE version to
equal the same RKHS function.
This is permissible because conditional means are
determined only $P_X$-almost surely.

The remaining class bounds can also be fixed independently
of $n$.
As verified for the local experiment, the margin condition
holds with constant one and treatment propensity is
one half.
Both potential outcomes lie in $[0,1]$, so
$\widetilde Y\in[-2,2]$.
Its centered conditional noise has range length at most
four and is therefore sub-Gaussian with variance proxy
at most four.
Consequently, the radius in Step~5, a margin bound
of at least one, and a noise-proxy bound of at least
four suffice to include both experiments.

\item We compare the upper rate with the unrestricted
minimax lower bound.
Both local experiments belong to the Gaussian-RKHS
class, so the worst-case risk over that class is at
least the worst-case risk over these two experiments.
By \Cref{thm:minimax}, every policy-learning algorithm
therefore faces a lower bound of order
$n^{-(\kappa+1)/(\kappa+2)}$ for $\rho\geq0$.

For fixed $\rho$, the variance term in Step~4 is
asymptotically smaller than its expected-regret term.
Their ratio is
$
\rho
\left\{
\frac{(\log n)^d}{n}
\right\}^{1/(\kappa+2)}
\rightarrow0.
$
The upper bound consequently has order at most
$
n^{-\frac{\kappa+1}{\kappa+2}}
(\log n)^{\frac{d(\kappa+1)}{\kappa+2}}.
$
Compared with the lower bound, the only discrepancy
is the factor
$(\log n)^{d(\kappa+1)/(\kappa+2)}$.
This proves the stated minimax comparison.
\end{enumerate}
\end{proof}

\section{Proofs for Appendix A}\label{app:secA_proofs}

\begin{proof}[Proof of \Cref{thm:joint_concavity}]
We proceed in steps.
\begin{enumerate}
\item We show that transformed welfare remains affine
in treatment probabilities.
Conditional on a population member's covariates and
potential outcomes, the treatment randomizer selects
$Y(1)$ with probability $G(X)$ and $Y(0)$ otherwise.
It therefore averages the two transformed outcomes:
\[
\begin{aligned}
V_\phi(G)
&=
\E\left[
G(X)\phi\{Y(1)\}
+\{1-G(X)\}\phi\{Y(0)\}
\right]\\
&=
\E[\phi\{Y(0)\}]
+
\E_X\{\tau_\phi(X)G(X)\}.
\end{aligned}
\]
The first term does not depend on the policy.
At each covariate value, maximizing the second term
over $G(X)\in[0,1]$ gives
$G_\phi^\star(X)=\1\{\tau_\phi(X)\geq0\}$.
This identity is exact and requires no approximation
of $\phi$.

\item We establish the welfare range needed for the
certainty-equivalent results.
Transformed-outcome integrability implies
$
\E_X|\tau_\phi(X)|
\leq
\E|\phi\{Y(1)\}|+\E|\phi\{Y(0)\}|
<\infty.
$
Moreover, Step~1 gives
\[
\begin{aligned}
0
&\leq V_\phi(G_\phi^\star)-V_\phi(G)
=
\E_X\left[
\tau_\phi(X)\{G_\phi^\star(X)-G(X)\}
\right]
\leq
\lVert\tau_\phi\rVert_{L_1(P_X)}.
\end{aligned}
\]
Thus all feasible transformed welfare values lie in
one finite interval.
Its endpoints are attained by $G_\phi^\star$ and
$1-G_\phi^\star$, so its length is
$\lVert\tau_\phi\rVert_{L_1(P_X)}$.

The certainty-equivalent arguments now apply to
$Z_A=V_\phi\{A(S)\}$.
Under the conditions of \Cref{lem:taylor} on an open
interval containing this welfare range, the same
Taylor expansion and remainder bound hold.
The optimizer and stability-based approximation results
retain their original uniformity conditions.
The exact CARA arguments also apply, using the range
just established in place of
$\lVert\tau\rVert_{L_1(P_X)}$.

\item We transfer the stability bounds to transformed welfare.
For any two policies $G$ and $G'$, the welfare identity
from Step~1 gives
\[
\begin{aligned}
|V_\phi(G)-V_\phi(G')|
&=
\left|
\E_X\left\{
\tau_\phi(X)\{G(X)-G'(X)\}
\right\}
\right|\\
&\leq
\E_X\left\{
|\tau_\phi(X)|\,|G(X)-G'(X)|
\right\}\\
&\leq
\lVert\tau_\phi\rVert_{L_1(P_X)}
\sup_{x\in\cX}|G(x)-G'(x)|.
\end{aligned}
\]
Apply this inequality to the policies learned before
and after replacing one experimental observation.
Taking the supremum over samples and replacements gives
$\beta_n(A)\leq
\lVert\tau_\phi\rVert_{L_1(P_X)}\Gamma_n(A)$,
where $\beta_n(A)$ is now defined using transformed
welfare and $\Gamma_n(A)$ retains its original definition.
This proves the transformed version of
\Cref{lem:policy_to_welfare}; its score-stability
implication follows from the same Lipschitz bound.

The proof of \Cref{thm:stability} uses independence
of the experimental observations and the bound on how
much replacing one observation changes welfare.
Applying that proof to $Z_A=V_\phi\{A(S)\}$, with
$\beta_n(A)$ defined as above, gives the corresponding
variance, risk-aware-regret, and tail bounds.
The exponential-moment bound in
\Cref{lem:mcdiarmid_welfare} applies for the same reason,
so \Cref{prop:cara_stability} also holds for transformed
welfare.

\item We transfer the bagging arguments and their rates.
The welfare loss of a bag-specific policy is
$
V_\phi(G_\phi^\star)-V_\phi(\widehat G_I)
=
\E_X\left[
|\tau_\phi(X)|
\1\{\widehat G_I(X)\neq G_\phi^\star(X)\}
\right].
$
On the disagreement event, the estimated transformed
effect must reach or cross zero.
Consequently,
$
|\tau_\phi(X)|
\leq
|\widehat\tau_{\phi,I}(X)-\tau_\phi(X)|.
$
These are the two identities used in the proof of
\Cref{thm:bag_upper}.
That proof therefore bounds expected transformed
welfare loss by
$Q_{\phi,m,2}^{(\kappa+1)/(\kappa+2)}$
and single-bag welfare variance by $Q_{\phi,m,2}$.

Complete averaging preserves the same variance contraction.
By affinity and symmetry,
$V_\phi\{A_m(S)\}$ is a complete $U$-statistic
with kernel $V_\phi(\widehat G_I)$.
Thus \Cref{lem:ustat} gives
\[
\V[V_\phi\{A_m(S)\}]
\leq
\frac mn\V\{V_\phi(\widehat G_I)\}
\leq
\frac mn Q_{\phi,m,2}.
\]
Combining the expected-loss and variance bounds proves
the complete-bagging result.

The conditional-mean and averaging identities also
remain unchanged.
Affinity supplies the mean equality, and concavity
of $\psi$ supplies the two Jensen inequalities in
\Cref{prop:bag_concave}.
The proof of \Cref{cor:finite_bagging} then gives the
finite-bag variance coefficient
$m/n+(1-m/n)/B$.
Substituting
$Q_{\phi,m,2}\lesssim m^{-\alpha}$ gives the stated
bound.

\item We embed the lower-bound construction in transformed
outcomes.
Suppose $\phi$ is strictly increasing on a nondegenerate
compact interval.
Choose a smaller compact interval inside its interior.
Concavity makes $\phi$ continuous there, so its image
is a nondegenerate interval on which $\phi^{-1}$
is well defined.

Start with the bounded potential outcomes $Y(d)$
from \Cref{thm:sharpness}.
Choose fixed constants $a>0$ and $b$ so that every
value of $b+aY(d)$ lies in this image.
Use
$\phi^{-1}\{b+aY(d)\}$ as the new potential outcomes.
They are bounded, and applying $\phi$ recovers exactly
$b+aY(d)$.

The learner can recover the original observed outcome
from the new one by applying $\phi$, subtracting $b$,
and dividing by $a$.
Apply the original base learner to these recovered
observations and multiply its score by $a$.
Then
$
\tau_\phi=a\tau,
$ $
\widehat\tau_{\phi,I}=a\widehat\tau_I,
$ and $
Q_{\phi,m,p}=a^pQ_{m,p}.
$
Because $a>0$, thresholding produces the same treatment
decisions as originally.

The welfare and margin bounds change only by fixed
scale factors.
For corresponding policies,
$V_\phi(G)=b+aW(G)$.
Expected regret therefore scales by $a$, and welfare
variance scales by $a^2$.
In particular,
$
R_{\phi,\RA}(A_m)
=
aR(A_m)
+
\rho a^2\V[W\{A_m(S)\}].
$
Also,
$
\Pp_X\{0<|\tau_\phi(X)|\leq t\}
=
\Pp_X\{0<|\tau(X)|\leq t/a\}
\leq
C_\tau a^{-\kappa}t^\kappa.
$
Both lower-bound terms from \Cref{thm:sharpness}
therefore retain their orders, with constants independent
of $n$, $m$, and $\rho$.
The construction belongs to the transformed high-level
class whenever its fixed margin, moment, and other
bounds admit these rescalings.

The same invertible transformation preserves the total
variation distance between the two sample laws in the
proof of \Cref{thm:minimax}, while multiplying expected
regret by $a$.
The lower bound therefore carries over, and the same
sample-independent fractional policy gives the matching
upper bound because its transformed welfare still has
zero sampling variance.
\end{enumerate}
\end{proof}


\begin{thebibliography}{99}

\bibitem[Andrews(1986)]{andrews_1986}
D. W. K. Andrews.
\newblock Stability comparisons of estimators.
\newblock \emph{Econometrica}, 54(5):1207--1235, 1986.

\bibitem[Andrews and Chen(2025)]{andrews_chen_2025}
I. Andrews and J. Chen.
\newblock Certified decisions.
\newblock arXiv:2502.17830, 2025.

\bibitem[Athey and Wager(2021)]{athey_wager_2021}
S. Athey and S. Wager.
\newblock Policy learning with observational data.
\newblock \emph{Econometrica}, 89(1):133--161, 2021.

\bibitem[Audibert and Tsybakov(2007)]{audibert_tsybakov_2007}
J.-Y. Audibert and A. B. Tsybakov.
\newblock Fast learning rates for plug-in classifiers.
\newblock \emph{The Annals of Statistics}, 35(2):608--633, 2007.

\bibitem[Belkin(2018)]{belkin_2018}
M. Belkin.
\newblock Approximation beats concentration? An approximation view on inference with smooth radial kernels.
\newblock In \emph{Proceedings of the 31st Conference on Learning Theory}, volume 75, pages 1348--1361, 2018.

\bibitem[Bousquet and Elisseeff(2002)]{bousquet_elisseeff_2002}
O. Bousquet and A. Elisseeff.
\newblock Stability and generalization.
\newblock \emph{Journal of Machine Learning Research}, 2:499--526, 2002.

\bibitem[Breiman(1996)]{breiman_1996}
L. Breiman.
\newblock Bagging predictors.
\newblock \emph{Machine Learning}, 24:123--140, 1996.

\bibitem[Breiman(2001)]{breiman_2001}
L. Breiman.
\newblock Random forests.
\newblock \emph{Machine Learning}, 45:5--32, 2001.

\bibitem[B\"uhlmann and Yu(2002)]{buhlmann_yu_2002}
P. B\"uhlmann and B. Yu.
\newblock Analyzing bagging.
\newblock \emph{The Annals of Statistics}, 30(4):927--961, 2002.

\bibitem[Buja and Stuetzle(2006)]{buja_stuetzle_2006}
A. Buja and W. Stuetzle.
\newblock Observations on bagging.
\newblock \emph{Statistica Sinica}, 16(2):323--351, 2006.

\bibitem[Chen et al.(2022)]{chen_syrgkanis_austern_2022}
Q. Chen, V. Syrgkanis, and M. Austern.
\newblock Debiased machine learning without sample-splitting for stable estimators.
\newblock In \emph{Advances in Neural Information Processing Systems}, volume 35, pages 3096--3109, 2022.

\bibitem[Chernozhukov et al.(2026)]{chernozhukov_newey_singh_syrgkanis_2026}
V. Chernozhukov, W. K. Newey, R. Singh, and V. Syrgkanis.
\newblock Adversarial estimation of Riesz representers.
\newblock \emph{Journal of the American Statistical Association}, 121(554):1398--1409, 2026.

\bibitem[Chernozhukov et al.(2025)]{chernozhukov_lee_rosen_sun_2025}
V. Chernozhukov, S. Lee, A. M. Rosen, and L. Sun.
\newblock Policy learning with confidence.
\newblock arXiv:2502.10653, 2025.

\bibitem[Denti and Pomatto(2022)]{denti_pomatto_2022}
T. Denti and L. Pomatto.
\newblock Model and predictive uncertainty: A foundation for smooth ambiguity preferences.
\newblock \emph{Econometrica}, 90(2):551--584, 2022.

\bibitem[Efron and Stein(1981)]{efron_stein_1981}
B. Efron and C. Stein.
\newblock The jackknife estimate of variance.
\newblock \emph{The Annals of Statistics}, 9(3):586--596, 1981.

\bibitem[Elisseeff et al.(2005)]{elisseeff_evgeniou_pontil_2005}
A. Elisseeff, T. Evgeniou, and M. Pontil.
\newblock Stability of randomized learning algorithms.
\newblock \emph{Journal of Machine Learning Research}, 6:55--79, 2005.

\bibitem[Ghirardato et al.(2003)]{ghirardato_et_al_2003}
P. Ghirardato, F. Maccheroni, M. Marinacci, and M. Siniscalchi.
\newblock A subjective spin on roulette wheels.
\newblock \emph{Econometrica}, 71(6):1897--1908, 2003.

\bibitem[Hampel(1974)]{hampel_1974}
F. R. Hampel.
\newblock The influence curve and its role in robust estimation.
\newblock \emph{Journal of the American Statistical Association}, 69(346):383--393, 1974.

\bibitem[Harsanyi(1953)]{harsanyi_1953}
J.~C.~Harsanyi.
\newblock Cardinal utility in welfare economics and in the theory of risk-taking.
\newblock \emph{Journal of Political Economy}, 61(5):434--435, 1953.

\bibitem[Hoeffding(1948)]{hoeffding_1948}
W. Hoeffding.
\newblock A class of statistics with asymptotically normal distribution.
\newblock \emph{The Annals of Mathematical Statistics}, 19(3):293--325, 1948.

\bibitem[Kitagawa et al.(2026)]{kitagawa_lee_qiu_2026}
T. Kitagawa, S. Lee, and C. Qiu.
\newblock Treatment choice with nonlinear regret.
\newblock \emph{Biometrika}, 113(2):asag008, 2026.

\bibitem[Kitagawa and Tetenov(2018)]{kitagawa_tetenov_2018}
T. Kitagawa and A. Tetenov.
\newblock Who should be treated? Empirical welfare maximization methods for treatment choice.
\newblock \emph{Econometrica}, 86(2):591--616, 2018.

\bibitem[Kitagawa and Tetenov(2021)]{kitagawa_tetenov_2021}
T. Kitagawa and A. Tetenov.
\newblock Equality-minded treatment choice.
\newblock \emph{Journal of Business \& Economic Statistics}, 39(2):561--574, 2021.

\bibitem[Klibanoff et al.(2005)]{klibanoff_et_al_2005}
P. Klibanoff, M. Marinacci, and S. Mukerji.
\newblock A smooth model of decision making under ambiguity.
\newblock \emph{Econometrica}, 73(6):1849--1892, 2005.

\bibitem[Liu and Molinari(2024)]{liu_molinari_2025}
Y. Liu and F. Molinari.
\newblock Inference for an algorithmic fairness-accuracy frontier.
\newblock arXiv:2402.08879, 2024.

\bibitem[Maccheroni et al.(2006)]{maccheroni_marinacci_rustichini_2006}
F. Maccheroni, M. Marinacci, and A. Rustichini.
\newblock Ambiguity aversion, robustness, and the variational representation of preferences.
\newblock \emph{Econometrica}, 74(6):1447--1498, 2006.

\bibitem[Maccheroni et al.(2013)]{maccheroni_marinacci_ruffino_2013}
F. Maccheroni, M. Marinacci, and D. Ruffino.
\newblock Alpha as ambiguity: Robust mean-variance portfolio analysis.
\newblock \emph{Econometrica}, 81(3):1075--1113, 2013.

\bibitem[Manski(2004)]{manski2004}
C. F. Manski.
\newblock Statistical treatment rules for heterogeneous populations.
\newblock \emph{Econometrica}, 72(4):1221--1246, 2004.

\bibitem[Manski and Tetenov(2007)]{manski_tetenov_2007}
C. F. Manski and A. Tetenov.
\newblock Admissible treatment rules for a risk-averse planner with experimental data on an innovation.
\newblock \emph{Journal of Statistical Planning and Inference}, 137(6):1998--2010, 2007.

\bibitem[Manski and Tetenov(2023)]{manski_tetenov_2023}
C. F. Manski and A. Tetenov.
\newblock Statistical decision theory respecting stochastic dominance.
\newblock \emph{The Japanese Economic Review}, 74(4):447--469, 2023.

\bibitem[Markowitz(1952)]{markowitz1952}
H. Markowitz.
\newblock Portfolio selection.
\newblock \emph{The Journal of Finance}, 7(1):77--91, 1952.

\bibitem[Mbakop and Tabord-Meehan(2021)]{mbakop_tabord_meehan_2021}
E. Mbakop and M. Tabord-Meehan.
\newblock Model selection for treatment choice: Penalized welfare maximization.
\newblock \emph{Econometrica}, 89(2):825--848, 2021.

\bibitem[McDiarmid(1989)]{mcdiarmid1989}
C. McDiarmid.
\newblock On the method of bounded differences.
\newblock In J. Siemons, editor, \emph{Surveys in Combinatorics, 1989}, volume 141 of London Mathematical Society Lecture Note Series, pages 148--188. Cambridge University Press, 1989.

\bibitem[Mentch and Hooker(2016)]{mentch_hooker_2016}
L. Mentch and G. Hooker.
\newblock Quantifying uncertainty in random forests via confidence intervals and hypothesis tests.
\newblock \emph{Journal of Machine Learning Research}, 17(26):1--41, 2016.

\bibitem[Moon(2025)]{moon_2026}
S. Moon.
\newblock Optimal policy choices under uncertainty.
\newblock arXiv:2503.03910, 2025.

\bibitem[Mourtada and Rosasco(2022)]{mourtada_rosasco_2022}
J. Mourtada and L. Rosasco.
\newblock An elementary analysis of ridge regression with random design.
\newblock \emph{Comptes Rendus. Math\'ematique}, 360(G9):1055--1063, 2022.

\bibitem[Opocher(2026)]{opocher_26_review}
G. Opocher.
\newblock Producing policy recommendations: From statistical decision theory to empirical practice.
\newblock arXiv:2607.29281, 2026.

\bibitem[Popoviciu(1935)]{popoviciu1935}
T.~Popoviciu.
\newblock Sur les \'equations alg\'ebriques ayant toutes
leurs racines r\'eelles.
\newblock \emph{Mathematica (Cluj)}, 9:129--145, 1935.

\bibitem[Qian et al.(2025)]{qian_ying_lam_yin_2025}
H. Qian, D. Ying, H. Lam, and W. Yin.
\newblock Subsampled ensemble can improve generalization tail exponentially.
\newblock In \emph{Advances in Neural Information Processing Systems}, volume 38, pages 5137--5188, 2025.

\bibitem[Rawls(1971)]{rawls_1971}
J.~Rawls.
\newblock \emph{A Theory of Justice}.
\newblock Belknap Press of Harvard University Press,
Cambridge, MA, 1971.

\bibitem[Shao(1999)]{shao1999}
J.~Shao.
\newblock \emph{Mathematical Statistics}.
\newblock Springer, New York, 1999.

\bibitem[Soloff et al.(2024a)]{soloff_barber_willett_2024}
J. A. Soloff, R. F. Barber, and R. Willett.
\newblock Bagging provides assumption-free stability.
\newblock \emph{Journal of Machine Learning Research}, 25(131):1--35, 2024a.

\bibitem[Soloff et al.(2024b)]{soloff_barber_willett_argmax_2024}
J. A. Soloff, R. F. Barber, and R. Willett.
\newblock Building a stable classifier with the inflated argmax.
\newblock In \emph{Advances in Neural Information Processing Systems}, volume 37, pages 70349--70380, 2024b.

\bibitem[Strzalecki(2011)]{strzalecki_2011}
T. Strzalecki.
\newblock Axiomatic foundations of multiplier preferences.
\newblock \emph{Econometrica}, 79(1):47--73, 2011.

\bibitem[Strzalecki(2013)]{strzalecki_2013}
T. Strzalecki.
\newblock Temporal resolution of uncertainty and recursive models of ambiguity aversion.
\newblock \emph{Econometrica}, 81(3):1039--1074, 2013.

\bibitem[Sun(2026)]{sun2026}
L. Sun.
\newblock Empirical welfare maximization with constraints.
\newblock \emph{Journal of Econometrics}, 253:106169, 2026.

\bibitem[Sun and Xiang(2025)]{sun_xiang_2025}
Y.~Sun and D.-H.~Xiang.
\newblock Total stability of outcome weighted learning.
\newblock \emph{Mathematical Foundations of Computing},
  8(5):734--755, 2025.

\bibitem[Swaminathan and Joachims(2015)]{swaminathan_joachims_2015}
A. Swaminathan and T. Joachims.
\newblock Counterfactual risk minimization: Learning from logged bandit feedback.
\newblock In \emph{Proceedings of the 32nd International Conference on Machine Learning}, volume 37, pages 814--823, 2015.

\bibitem[Yang et al.(2017)]{yang_pilanci_wainwright_2017}
Y. Yang, M. Pilanci, and M.~J. Wainwright.
\newblock Randomized sketches for kernels: Fast and optimal nonparametric regression.
\newblock \emph{The Annals of Statistics}, 45(3):991--1023, 2017.

\bibitem[van der Laan(2026)]{van_der_laan_2026}
L. van der Laan.
\newblock A researcher's guide to empirical risk minimization.
\newblock arXiv:2602.21501, 2026.

\bibitem[Vickrey(1945)]{vickrey_1945}
W.~Vickrey.
\newblock Measuring marginal utility by reactions to risk.
\newblock \emph{Econometrica}, 13(4):319--333, 1945.

\bibitem[Viviano and Bradic(2024)]{viviano_bradic_2024}
D. Viviano and J. Bradic.
\newblock Fair policy targeting.
\newblock \emph{Journal of the American Statistical Association}, 119(545):730--743, 2024.

\bibitem[Wager and Athey(2018)]{wager_athey_2018}
S. Wager and S. Athey.
\newblock Estimation and inference of heterogeneous treatment effects using random forests.
\newblock \emph{Journal of the American Statistical Association}, 113(523):1228--1242, 2018.

\bibitem[Wald(1950)]{wald1950}
A. Wald.
\newblock \emph{Statistical Decision Functions}.
\newblock John Wiley \& Sons, New York, 1950.

\end{thebibliography}
\end{document}